\documentclass[11pt] {article}% use larger type; default would be 10ptpts
\usepackage{amsmath,amssymb,amsfonts}
\usepackage[utf8]{inputenc} % set input encoding (not needed with XeLaTeX)
\usepackage{amsthm}

\usepackage{tikz}
\usepackage{hyperref}
\usepackage{graphicx}
\usepackage{amsfonts}

\usepackage{hyperref}

\usepackage{color}
\usepackage[english]{babel}
\usepackage{ifthen}
\usepackage{graphicx}
\usepackage{tikz}
\usetikzlibrary{decorations.markings}
\usetikzlibrary{plotmarks}
\usetikzlibrary{patterns}
\usepackage{pgfplots}
\usepgfplotslibrary{fillbetween}
\usepackage{makeidx}
\usetikzlibrary{calc}

\usepackage{geometry} % to change the page dimensions
\usepackage{graphicx} 
\usepackage{array} % for better arrays (eg matrices) in maths
\usepackage{paralist} % very flexible & customisable lists (eg. enumerate/itemize, etc.)
\usepackage{verbatim} % adds environment for commenting out blocks of text & for better verbatim
\usepackage{subfig} % make it possible to include more than one captioned figure/table in a single float

\usepackage{textpos}

\usepackage[backend=bibtex,style=numeric]{biblatex}            
\tikzset { domaine/.style 2 args={domain=#1:#2} }

\newcommand{\KN}{\mathbin{\bigcirc\mspace{-15mu}\wedge\mspace{3mu}}}

\newtheorem{theorem}{Theorem}

\newtheorem{theo}{Theorem}[section]
\newtheorem*{theo*}{Theorem}

\newtheorem{prop}[theo]{Proposition}
\newtheorem*{prop*}{Proposition}
\newtheorem{lem}{Lemma}[section]
\newtheorem{cor}{Corollary}[section]
\newtheorem*{cor*}{Corollary}
\newtheorem{de}{Definition }[section]

\newtheorem{remark}{Remark}[section]

\newcommand{\nocontentsline}[3]{}
\newcommand{\tocless}[2]{\bgroup\let\addcontentsline=\nocontentsline#1{#2}\egroup}

\title{On the Rigidity of Cosmological Space-Times}
\author{Rodrigo Avalos \\ Email: \href{rdravalos@gmail.com}{rdravalos@gmail.com} \footnote{Universität Potsdam, Institut für Mathematik, Potsdam, Deutschland.}}
\date{}

\begin{document}
\maketitle
%\tableofcontents

\begin{abstract}
In this paper we analyse a family of geometrically well-behaved cosmological space-times $(V^{n+1},g)$, which are foliated by \emph{intrinsically isotropic} space-like hypersurfaces $\{M_t\}_{t\in \mathbb{R}}$, which are orthogonal to a family of co-moving observers defined by a global time-like vector field $U$. In particular, this implies such space-times satisfy several of the well-known criteria for isotropic cosmological space-times, although, in the family in question, the simultaneity-spaces $(M_t,g_t)$ associated to $U$ can have as sectional curvature a sign-changing function $k(t)$. Being this clearly impossible in the FLRW family of standard cosmological space-times, it motivates us to revisit the geometric rigidity consequences of different definitions of isotropy available in the literature. In this analysis, we divide such definition according to whether the isometries involved are taken to be (local) space-time (STI space-times) or (local) space isometries (SI space-times) of $(M_t,g_t)$ for each $t$. This subtlety will be shown to be critical, proving that only when space-time isometries are considered one obtains the well-known rigidity properties associated with isotropic cosmological space-times. In particular SI space-times will be shown to be a strictly larger class than the STI ones, allowing a family of basic cosmological curvature change models which are not even locally isometric to any FLRW space-time.
\end{abstract}

\section{Introduction}

This paper is motivated by a family examples of well behaved Lorentzian manifolds, which within the context of general relativity (GR) can be regarded as \emph{cosmological space-times}, recently introduced in \cite{MiguelCosmologicalST}. These are globally hyperbolic space-times $(V^{n+1},g)$, foliated by hypersurfaces $\{M_t\}_{t\in I}$, $I\subset \mathbb{R}$, orthogonal to a global time-like vector field $U=\partial_t$, where $M_t$ is (for each $t$) isometric to a (Riemannian) space of constant sectional curvature $M_{k(t)}$, where $k:I\mapsto \mathbb{R}$ denotes the corresponding sectional curvature function. One of the remarkable features about these examples is that the sectional curvature function can change sign, and therefore we refer to these models as \emph{basic cosmological curvature change models} (BCCCM). What makes this feature specially remarkable, is that it puts these examples at odds with some very standard claims about rigidity of cosmological space-times. To be more precise, these examples put into evidence that several well-known notions of \emph{isotropy} commonly used in GR are inequivalent, and not all of them lead to the well-known rigidity properties of isotropic cosmological space-times. That is, the claim that an everywhere isotropic cosmological space-time must be (locally) isometric to a standard Friedman-Lemaître-Robertson-Walker (FLRW) space-time is highly dependant on the chosen definition of isotropy among several well-known and highly standard ones, such as the ones used in \cite{BlauGRLN,CarrollBook,CB-Book,MTW,ONeillBook,CarrollLN,WaldBook,Weinberg}.\footnote{Even more directly, some of these claims actually entail a constructive way to obtain (uniquely) a space-time within the FLRW family from the chosen definition of isotropy. The BCCCM family is, for instance, a direct counterexample to some of these claims (see Remark \ref{NonRigRemark1}).}

In order to provide some background on the relevance of the results presented in this paper, our starting point is the
%It is a 
remarkable fact that, in the case of cosmology, there seems to be compelling evidence to believe that the Universe, in very large scales, is highly symmetric. In particular, that it is nearly \emph{isotropic}. The main evidence for this comes from the study of cosmic microwave background radiation. Thus, the starting point in the standard model of cosmology is the assumption that this symmetry is exact, so as to determine solutions approximating the problem, and then study perturbations of these solutions to understand the deviations of isotropy that we actually observe. For this program it is certainly beneficial that isotropic space-times are sufficiently well-characterised within a sufficiently small family of solutions which can be studied together. In the standard model of cosmology, the associated family is the FLRW models, which are given by space-times $(V,g)$, with $V=I\times M$, $I$ an interval in $\mathbb{R}$ parametrised by a coordinate $t$; $M$ a $3$-dimensional manifold; $g=-dt^2+a^2(t)\gamma_{\epsilon}$, where %$\gamma_{\epsilon}$ is a fixed Riemannian metric on $M$ with constant sectional curvature $k=\epsilon$. Therefore, the mathematical model gets reduced to a warped product between $(I,-dt^2)$ and $(M,\gamma)$, which we compactly write $V=I\times_aM$. In the standard treatment of these solutions, one furthermore assumes that the \emph{space} $M$ is simply connected, leading to
\begin{align}\label{SpaceForms}
\begin{split}
M_{\epsilon}\doteq (M,\gamma_{\epsilon})=\begin{cases}
\mathbb{S}^3, \text{ if } \epsilon>0,\\
\mathbb{E}^3, \text{ if } \epsilon=0,\\
\mathbb{H}^3, \text{ if } \epsilon<0,
\end{cases}
\end{split}
\end{align}
and $a:I\mapsto\mathbb{R}^{+}$ is the warping function. Above, $\mathbb{S}^{3},\mathbb{E}^3$ and $\mathbb{H}^3$ denote (respectively) the standard $3$-sphere, Euclidean space and $3$-dimensional hyperbolic space. Therefore, the mathematical model gets reduced to a warped product between $(I,-dt^2)$ and $M_{\epsilon}$, which we compactly write $V=I\times_aM_{\epsilon}$. In the above scenario, the curves $t\mapsto (t,p)$ distinguish a family of time-like geodesics (representing \emph{observers}) for whom \emph{space/space-time looks isotropic}. The warping function is then determined by solving the Einstein equations, typically for a \emph{perfect fluid source}, which get reduced to the so called \emph{Friedman-Lemaître equations}.
%\begin{align}\label{FriedmanEqs}
%\begin{split}
%-3\frac{\ddot{a}}{a}&=\frac{1}{2}(\mu+3p)-\Lambda,\\
%3\left(\frac{\dot{a}}{a}\right)^2+3\frac{\epsilon}{a^2}&=\bar{\mu}  + \Lambda,\\
%\dot{\bar{\mu}}+3(\bar{\mu}+\bar{p})\frac{\dot{a}}{a}&=0.
%\end{split}
%\end{align}
%In the above system, the functions $\mu$ and $p$ stand for the \emph{proper} energy and pressure density characterising the perfect fluid, while $\Lambda$ stands for the \emph{cosmological constant}, and this system can be integrated explicitly for some simple \emph{state equations} relating $\bar{\mu}$ and $\bar{p}$, such as for dust ($\bar{p}=0$), radiation ($\bar{\mu}-3\bar{p}=0$) and also the $\Lambda$-vacuum cases which lead to deSitter and anti-deSitter solutions. Some of these cases model specific stages in the history of the universe. More importantly, under reasonable assumptions of $\bar{\mu}$ and $\bar{p}$, equations (\ref{FriedmanEqs}) are good enough to give us an overall picture of the dynamics of the system.
To a first approximation, this gives a good cosmological qualitative description. Details about this can be consulted in \cite{CarrollBook,WaldBook,Weinberg,WeinbergCosmology} and more mathematically oriented description of these space-times can be found in \cite{ONeillBook}. The rigidity associated to the FLRW family certainly simplifies the work of the cosmologist to the analysis of perturbations of it. Although this is not a necessary condition for the successful study of cosmology, the larger the idealised family of solutions is, the more complicated it is to single out the solutions that actually model observations. %Furthermore, the weaker the rigidity of the family, the more necessary observations are to single out the families which can actually model the system.
%In contrast, one may argue that the higher the rigidity of the family of approximate solutions, the more outstanding the results are when they do match observations, since there is a greater sense of uniqueness in the prediction of the theory.  

%All of the above highlights the importance within physics of rigidity properties associated to cosmological space-times. 
The claim that isotropic space-times exhibit the rigidity described above, and thus get reduced to the FLRW family, is a standard one in GR and cosmology, as one can see in classic and modern textbooks, lecture notes and research papers such as \cite{BlauGRLN,CarrollBook,CB-Book,ClarksonMaartens,Maartens01,MTW,CarrollLN,WaldBook,Weinberg,WeinbergCosmology}. It is thus interesting that the definition of isotropy itself is not uniform within such literature. Let us first comment the basic structures imposed on isotropic cosmological space-times $(V^{n+1},g)$ which are common to most definitions. To start with, a common point of agreement is to assume the existence of a foliation of $V^{n+1}$ by space-like hypersufaces $M_t$ which admit a time-like unit normal field $U=\partial_t$ whose integral curves represent a family of observers for whom space-time is isotropic. Thus, one assumes that $V^{n+1}\cong I\times M^n$, $I\subset \mathbb{R}$, $\bar{g}(U,U)=-1$, and we denote by $M^n_t= \{t\}\times M^n$. The \emph{model} for these kind of space-times therefore starts with a space-time metric of the form
\begin{align}\label{CosmologicalMetric1}
g=-dt^2+g_t,
\end{align} 
where $g_t$ is a $(0,2)$-tensor field inducing a time-dependent set of Riemannian metrics on $M^n_t$. One final point of agreement among different definitions of isotropy is that one is concerned about \emph{space-isotropy} as seen by the isotropic observers defined by $U$. We will refer to a space-time which satisfies all the above generally agreed notions as a \emph{cosmological space-time}.\footnote{Notice that, up to this point, a cosmological space-time is defined as a globally hyperbolic space-time for which its \emph{lapse function} is set to one and its \emph{shift vector} to zero, on the a priori basis of \emph{isotropy}.}

Let us now disentangle the definitions of isotropy found in the previous references into two families. On the one hand, some authors characterise isotropy in terms of the existence of symmetries of each space-like slice $(M^n_t,g_t)$ for all $t$, which translates to invariance under rotations in tangent spaces to $M^n_t$. We shall refer to this family as \emph{space-isotropic} (SI) space-times, and we highlight references such as \cite{CarrollBook,CB-Book} which appeal to this criterion.\footnote{See Section \ref{DefsClassifications} for details on these classification.} On the other hand, other authors appeal to the existence of space-time symmetries around any point which leave invariant the vector field $U$ and act as point-wise rotations on $U^{\perp}$. References adopting this second view are, for instance, \cite{ONeillBook,WaldBook,Weinberg}, and we shall refer to it as \emph{space-time isotropy} (STI) definitions. The subtle difference here is whether the isometries around any point $(t,p)$ in a cosmological space-time are isometries of the space-time metric $g$ or the space-metrics $g_t$ (for each $t\in I$). In both cases, the classic claim is that one is lead to the same rigidity: the only complete and simply connected models satisfying the corresponding definition of isotropy are within the FLRW family. The main purpose of this paper is to show that this is actually only true under STI definitions, and that SI definitions are strictly weaker that STI definitions, since there are (well-behaved) Lorentzian manifolds satisfying SI definitions which are non-isometric (even locally) to any FLRW space-time. 

The basis for the above claim is the BCCCM family of space-times recently introduced in \cite{MiguelCosmologicalST}. These are warped-product Lorentzian manifolds, which we denote by $V^{n+1}_{k(t)}$, and have the form $P_{+}\times_{S}\mathbb{S}^{n-1}$, where $P_{+}$ denotes the half-plane $\{(t,r)\in \mathbb{R}^2\: :\: r\geq 0\}$ and the warping function $S:P_{+}\to \mathbb{R}^{+}$ is given in terms of classic solutions to the Jacobi equation.\footnote{See Section \ref{Preliminaries} for further details.} For instance, a simple example of such space-times is given by taking
\begin{align}\label{MiguelsMetric}
g\doteq -dt^2 + dr^2 + \sinh^2(tr)g_{\mathbb{S}^{n-1}}=-dt^2+g_t.
\end{align}
The above space-time satisfies all the requirements of SI definitions, since $(M_t,g_t)$ are (complete and simply connected) Riemannian manifolds invariant under rotations of $\mathbb{S}^{n-1}$, which translates into a point wise $SO(n)$ symmetry. These spaces $(M^n_t,g_t)$ are in particular, for each time $t$, isometric to the standard hyperbolic space $\mathbb{H}_{k(t)}^{n}$ of constant sectional curvature $k(t)=-t^2$. As shown in \cite{MiguelCosmologicalST}, the associated space-times are globally hyperbolic with the level sets $t=cte$ as Cauchy surfaces. This simple example belongs to a wider family (which we describe in Section \ref{Preliminaries}) of globally hyperbolic space-times that satisfy the SI definition of isotropy. Notice that these examples allow the sectional curvature function $k:I\mapsto\mathbb{R}$ to change sign smoothly, which is clearly impossible in the FLRW case.%In this wider family, the warping function is taken to be $\mathrm{sn}_{k(t)}(r)$, where $k:\mathbb{R}\mapsto\mathbb{R}$ is an arbitrary smooth function. We will denote this family as the \emph{basic cosmological curvature change models} (BCCCM).

An interesting point to note here, is that the above family of space-times seems to \emph{violate} the construction of the FLRW family provided in some textbook references, since it can be constructed by taking $U=\partial_t$ as a unit time-like vector, whose integral curves are geodesics, which are orthogonal to the spaces $(M^n_t,g_t)$ which foliate the space-time (as Cauchy surfaces) and, furthermore, each of these Cauchy surfaces is a complete and simply connected space-form. Nevertheless, the resulting space-time $V^{n+1}_{k(t)}$ is not of the form of a FLRW space-time. Although this alone shows that rigidity associated with the \emph{construction} following SI definitions of isotropy does not hold, one can pose an interesting and related geometric question, which is whether these kinds of examples could actually be isometric to some standard FLRW space-time. After all, there is in general no unique family of isotropic observers, and therefore the space-time could have the form of a FLRW for a conjectured second family of isotropic observers. If this were true, one could expect that the SI definition of isotropy, in an appropriate sense, still leads uniquely to the FLRW family. Nevertheless, along this paper we will prove the following result:

\begin{theorem}\label{ThmAIntro}
A BCCCM is (locally) isometric to a FLRW space-time iff the sectional curvature function $k:I\subset \mathbb{R}\mapsto\mathbb{R}$ is constant, in which case the corresponding BCCCM is isometric to $-I\times M^n_k$.
\end{theorem}
%Along this paper, we will show that this not possible and that the examples provided by \cite{MiguelCosmologicalST} are non-isometric to any FLRW space-time. This provides a whole new family of distinct SI space-times, which do not belong to the FLRW family. 
The proof of the above result resides in studying conformal properties of the BCCCM family and showing they are conformally flat only in the rigid cases highlighted in the above theorem. This result stand out because of how standard the treatment of isotropic space-times is within physics. Let us contrast with the fact that STI space-times do have the well-known rigidity for cosmological space-times, although a complete and clear proof of this statement does not seem to be easy to find, nor is it cited within standard literature in physics. In particular, we claim the rigidity follows from \cite[Chapter 12, Proposition 6]{ONeillBook} and arguments laid out after it.\footnote{See Theorem \ref{ONthm} in Section \ref{Preliminaries}.} We shall elaborate further on a few of these subtleties within Section \ref{Preliminaries}, but at this point we would like to highlight that the realisation that not all definitions of isotropy within cosmology are equivalent, and furthermore that there may be small gaps within certain standard arguments, leads to a follow-up natural problem, which is distinguishing those definitions which are in fact equivalent and do imply rigidity with the FLRW family. %To this end, we shall first explain how one STI definition of isotropy does lead to this rigidity (Theorem \ref{ONthm}), and then Theorem \ref{ThmAIntro} shows that SI definitions are strictly weaker, allowing for instance a whole second family of non-isometric solutions. Concerning \emph{different} STI definitions, we shall concentrate in a comparison between a well-known definition via the existence of local space-time symmetries fixing $U$ and acting as point wise rotations on $U^{\perp}$ (see Definition \ref{WaldIsotropy}) and definitions which appeal to the existence of special and large enough families of local Killing vector fields (see Definition \ref{WeinbergIsotropy}). In this case, the following holds:

With the above in mind, in Section \ref{DefsClassifications}, we shall introduce definitions of isotropy known from standard literature, which we extract in particular from \cite{CarrollBook,CB-Book,ONeillBook,WaldBook,Weinberg}, and which we believe are quite representative. Among these definitions, we will see that either in the SI or STI categories, one distinguishes authors who introduce isotropy explicitly by the existence of local isometries with special properties on each tangent space, while other introduce this notion through the existence of families of local Killing vector fields with specific point wise properties. We shall refer to the first kind as \emph{type I}-definitions and to second one as \emph{type II}-definitions. The equivalence of these approaches may seem quite intuitive and, actually, this seems to have been accepted quite broadly in the literature. Nevertheless, it does not seem to be easy to find self-contained explicit proofs with clear hypotheses for such equivalence.\footnote{Also, note that even though such equivalences seem to be intuitive, the actual proofs provided in Section \ref{DefsClassifications} are not as trivial as to assume one can do without them.} Putting this together with the above discussion on non-rigidity of SI definitions, we believe it to be important to provide such clear geometric proofs, and therefore in Section \ref{DefsClassifications} we shall prove the following:\footnote{See Theorem \ref{SIuniquenessTHM} and Theorem \ref{STIuniquenessTHM}.}
\begin{theorem}\label{ThmBIntro}
Given a cosmological space-time $(V^{n+1},g)$, then:
\begin{enumerate}
\item It is SI type I iff it is SI type II;
\item It is STI type I iff it is STI type II.
\end{enumerate}
\end{theorem}
Notice that in classic textbooks all the above definitions are claimed to lead uniquely to the same family of space-times. Thus, if one accepted the classical rigidity claims, the above theorem would become automatic, although a false equivalence claim between the STI and SI family would also become automatic. Thus, once we realise that (some of the) classical rigidity claims and arguments leading to them have gaps, it becomes problematic to commit to any such claim without proof checking. That is why we shall commit to the results of \cite[Chapter 12, Proposition 6]{ONeillBook}, put together in Theorem \ref{ONthm}, where the reader can check the validity of the claim quite easily. In particular, to the best of our knowledge, this is the most direct and geometrically transparent proof available. In contrast, the arguments in favour of the rigidity claim associated with the STI type II definition, which can be found in \cite{Weinberg}, are based on long tensor calculus manipulations, which are less geometric in nature, and furthermore the results are typically expressed only in a local manner with certain hypotheses not fully explicit. Therefore, instead of proof checking these other arguments, we prove that both definitions are in fact equivalent and therefore both lead to the desired rigidity properties through Theorem \ref{ONthm}. We also consider this path to be more fruitful, because the proof appeals to geometric constructions which explicitly link the more odd-looking properties demanded for the local family of Killing fields entering in type II definitions, with the properties of the local group of isometries (and their Lie algebras) entering in type I definitions. %These links seem to have been intuitively accepted in standard literature, although to the best of our knowledge no proof was available.

%Then, we shall contrast two well-known definitions of STI proving them be equivalent (Theorem XY). 
%In doing this, we shall take the arguments of \cite{ONeillBook} as our founding basis because, there, the notion of isotropy is well and compactly introduced, and there is a full proof of the rigidity statements associated with it. 
Following the above comments, let us also highlight a key issue which seems to have been sometimes overlooked to conclude rigidity with the FLRW family, which is the  
%Let us briefly provide a few further justification on why above we take the arguments of \cite{ONeillBook} as our default one for the proof of rigidity of cosmological space-times. Let us notice that other standard references which appeal to the same STI definition seem to neglect the 
necessity of establishing that the natural maps $\mu_{s,t}:M^n_s\to M^n_t$ on a cosmological space-time act by homotheties, before concluding that the fact that for each time $t$ the metric $g_t$ is isometric to simply-connected space-form implies that $g_t=a^2(t)\gamma_{\epsilon}$. As can be seen in \cite[Chapter 12, Proposition 6]{ONeillBook}, this statement needs a proof of its own, for which the existence of space-time isometries becomes essential. %In constant to other references which notice this fact, such as \cite{Weinberg}, the proof in \cite{ONeillBook} is quite concise and to the point. 
Notice that this action by homotheties is the crucial property which does not allow the $t=cte$ simultaneity spaces to change the sign of their (constant) sectional curvature. The fact that there exist SI isotropic space-times with sign-changing sectional curvature has been noticed in the past, for instance in \cite{Krasinski2,Krasinski1,Stephani}, where solutions related to the family studied in \cite{MiguelCosmologicalST} were analysed, and related studies are also being carried out in \cite{VeraMars}. %Nevertheless, this does not seem to have changed standard claims of rigidity. Related with these solutions, 
In this context, we believe that Theorems \ref{ThmAIntro} and \ref{ThmBIntro} call for a deeper discussion on which definition of isotropy should be physically favoured. %Thus, we take it as our default argument and we propose to compare other definitions and statements in the literature to it. We consider this to be a more fruitful program than proof checking each of the many different arguments appearing in the literature for the determination of the FLRW from a STI definition. Following this idea, in the paper we split the STI and SI showing the latter to be strictly weaker, and establish the equivalence of the STI definition used in \cite{Weinberg} to the one of \cite{ONeillBook}. We focus on the comparison with \cite{Weinberg} because in this case the definition of isotropy, although given in terms of space-times symmetries, is not (a priori) identical to the one of \cite{ONeillBook}. In the end, this establishes that the rigidity claims in \cite{Weinberg} do hold, although the arguments there are quite different in nature to the ones of \cite{ONeillBook}, in particular focusing on \emph{local} results.  

%We consider this analysis not only important within a mathematical program of classifying symmetric Lorentzian manifolds, but also relevant for current cosmology, since 
Let us finish this section highlighting that the results mentioned above bring about interesting physical questions, most notably, which type of definition should actually be physically favoured. The answer to this question would certainly rely on the way the evidence of isotropy is obtained and thus whether one should infer the existence of space-time symmetries or space symmetries. The determination of isotropy from actual physical measurements does not seem to be a completely settled issue, based on papers such as \cite{ClarksonMaartens,PerlickCosmology,Maartens01}. In those cases, it is particularly interesting that the authors highlight that, since the information we get comes from light travelling along our past light-cone, then the determination of isotropy from measurement should be based on geometric data on such light-cones. In the case of \cite{PerlickCosmology}, the authors propose to analyse \emph{anisotropies in the Hubble-law} in a power expansion, %which leads them to establish that if such law is isotropic to only a few orders, then isotropy must hold exactly and they obtain certain rigidity properties associated to \emph{third-order Hubble models}, which could have observational significance. 
although in this case the definition of an isotropic Hubble law seems to be of a different kind than the definitions of isotropy analysed in this paper. Perhaps more directly related to our discussion are some of the claims in \cite{ClarksonMaartens,Maartens01}, because there the authors claim that if certain \emph{initial data} on the past light cone for the characteristic Cauchy problem are \emph{isotropic}, then the whole interior of that past-light cone should also be \emph{isotropic}. Based on our analysis, we merely would like to point out that rigidity properties associated to such a statement will be highly dependent on the specific notion of isotropy one is putting forward. It should also be highlighted that the propagation of symmetries from an initial data set to space-time symmetries of the evolving space-times might be more subtle than what is described in these references. In the case of initial data on space-like Cauchy surfaces this relation is given by the link between KIDs and space-time Killing vectors, established in \cite{FMM,MoncriefSymmetries}.

With all of the above in mind, this paper shall be structured as follows. In Section \ref{Preliminaries}, we shall provide the preliminaries necessary for the main core of the paper. Then in Section \ref{DefsClassifications}, we shall classify the different notions of isotropy discussed above, and in particular prove Theorem \ref{ThmBIntro}. Section \ref{RigiditySection1} represents the core of this paper, where we shall establish a conformal characterisation of the BCCCM family and prove Theorem \ref{ThmAIntro}. For the benefit of the reader and to provide a more self-contained presentation, we also present an appendix with some well-known conformal properties of the FLRW family, for which is not easy to find explicit geometric proofs in the literature.
 
%In fact, we shall show that they are non-isometric to any element of a larger family of \emph{generalised} FLRW space-times.

%To do this, we shall below examine the conformal properties of space-times of the form of (\ref{MiguelsMetric}) and compare with those of FLRW space-times.

%space of constant sectional curvature $k=-t$. Furthermore, it satisfies the \emph{construction} implemented by \cite{CB-Book,CarrollBook} among others. That is $U=\partial_t$ is a unit time-like vector, the spaces $M_t$ orthogonal to $U$ foliate the space-time and each of these spaces satisfies the isotropy condition of Definition \ref{CBIsotropy}. Nevertheless, they are not of the form of a FLRW model. Although this alone shows that rigidity associated with the construction following Definition \ref{CBIsotropy} does not follow, one can pose a related geometric questions, which is whether these kinds of examples motivated by our previous sections could actually be isometric to some standard FLRW space-time. After all, there is in general no unique family of isotropic observers, and therefore the space-time could have the form of a FLRW for a conjectured second family of isotropic observers. Nevertheless, below we will show that this in general not possible and that the examples provided by (\ref{MiguelsMetric}) are non-isometric to any FLRW space-time. To do this, we shall below examine the conformal properties of space-times of the form of (\ref{MiguelsMetric}) and compare with those of FLRW space-times.

\section{Preliminaries}\label{Preliminaries}

\subsection{A canonical definition for isotropy}

%There are some basic structures imposed on isotropic cosmological space-times $(V^{n+1},g)$ which are common to most definitions. To be concise, let us notice that all these definitions assume the existence of foliation of $V^{n+1}$ by space-like hypersufaces $M_t$ which admit a time-like unit normal field $U=\partial_t$ whose integral curves represent a family of observers which \emph{sees} an isotropic space-time. Thus, a common starting point is to assume that $V^{n+1}\cong I\times M$, $I\subset \mathbb{R}$, $M_t\cong \{t\}\times M$, and furthermore that $\bar{g}(U,U)=-1$. Therefore, the \emph{model} for these kind of space-times starts with a space-time metric of the form
%\begin{align}\label{CosmologicalMetric1}
%g=-dt^2+g_t,
%\end{align} 
%where $g_t$ is the Riemannian metric induced on $M_t$ for each $t$. One final point of agreement among different definitions of isotropy is that one is concerned is \emph{space-isotropy} as seen by the isotropic observers. We will refer to a space-time which satisfies all the above generally agreed notions as a \emph{cosmological space-time}. At this point one could adopt one of the following definitions.

Appealing to the discussion of cosmological space-times given in the introduction, let us recall the following definition, as used in this paper.

\begin{de}[cosmological space-times]\label{CosmoligcalSTDef}
We will say that a Lorentzian manifold $(V^{n+1},g)$, $n\geq 3$, is a cosmological space-time if
\begin{enumerate}
\item $V^{n+1}=I\times M^n$, where $I$ is a connected open interval of $\mathbb{R}$, which we parametrise by some coordinate $t$;
\item $U\doteq \partial_t$ is orthogonal to $M^n$ and $g(U,U)=-1$.
\end{enumerate}
In this setting we denote by $M^n_t\doteq \{t\}\times M^n$, by $g_t$ the properly Riemannian induced metric on $M^n_t$ and by $\pi:V^{n+1}\to M$ the canonical projection on the second factor. Therefore, we can always write $g=-dt^2+g_t$.
\end{de}

Let us now present a highlighted definition for isotropy well-known from standard literature, having been used, for instance, in \cite{ONeillBook,WaldBook}, and which we shall take to be the \emph{default} STI definition. %First, let us present what we take to be the \emph{default} STI definition: 

\begin{de}[STI cosmological space-times - Type I]\label{WaldIsotropy}
A cosmological space-time is said to be isotropic if for every point $q\in V^{n+1}$, given unit vectors $v,w\in T_{q}V^{n+1}$ tangent to $M^n$, there exists an isometry $\psi=\mathrm{Id}\times \psi_M$, defined (at least) on a neighbourhood $\mathcal{U}\subset V^{n+1}$ of $q$, such that:
\begin{enumerate}
\item $d\psi_q(U)=U_q$;
\item $d\psi_{q}(v)=w$;
\item If $n>3$, we furthermore explicitly ask that given two planes $P_1,P_2\subset T_qV^{n+1}$ which are tangent to $M^n$ there is one such isometry $\psi$ satisfying the first item above, and such that $(d\psi)_q(P_1)=P_2$.
\end{enumerate}
\end{de}

%\begin{remark}
%Notice that the above definition implies that the map $\psi_M:M\mapsto M$ induces a local isometry of $(M_t,g_t=i^{*}g)$, where $i:M_t\mapsto V$ denotes the canonical inclusion $p\mapsto (p,t)$. Let $v\in T_pM$ be a unit vector, then
%\begin{align*}
%\psi_M^{*}g_t(v,v)=g_t(d\psi_Mv,d\psi_Mv)=g(d\psi_Mv,d\psi_Mv)=g(v,v)=g_t(v,v).
%\end{align*}
%Therefore, choosing and orthonormal basis $\{E_{i}\}_{i=1}^n$ of $T_pM$ to represent $(d\psi_M)_p:T_pM\mapsto T_pM$, this map becomes a linear isometry of Euclidean space, and must therefore be an orthogonal matrix.
%\end{remark}

%The above definition has been used for instance in \cite{ONeillBook,WaldBook}.

Before moving on, let us highlight that for $n=3$ the second condition in the above definition actually implies the third one, since any plane is uniquely determined as the orthogonal space to a normal vector. Thus, an isometry satisfying the second condition, also preserves the corresponding orthogonal spaces.

Within the main theme of this paper, we distinguish the above definition of space-time isotropy mainly because of two reasons. Firstly, because the intuitive notion of space-isotropy is naturally and transparently introduced, and secondly because it does lead to rigidity with FLRW family due to arguments available in the literature. Interestingly enough, to the best of our knowledge, these arguments have only been fully spelled out in \cite[Chapter 12]{ONeillBook}. There, the reader may find the following result, which we fully write down due to its relevance as motivation for this paper.

\begin{theo}[O'Neill]\label{ONthm}
Let $(V^{n+1},g)$ be an $(n+1)$-dimensional cosmological space-time which is isotropic under Definition \ref{WaldIsotropy}. Then, 
\begin{enumerate}
\item[i.] Each slice $M^n_t=(\{ t\}\times M^n,g_t)$ has constant sectional curvature $k(t)$;
\item[ii.] Given $s,t\in I$, the diffemorphism $\mu_{s,t}:(M^n_s,g_s)\to (M^n_t,g_t)$ defined by $\mu_{s,t}(p)=p$ is an homothety;
%\item For any given $s\in I$ fixed, defining $\gamma\doteq g_s$, it holds that $g_t=f^2(t)\gamma$ for a smooth function $f:I\mapsto \mathbb{R}^{+}$;
\item[iii.] $(V^{n+1},g)$ is globally isometric to a warped product $I\times_{f}M^n_{k}$, with $M^n_{k}=(M^n,\gamma_{k})$ a Riemannian manifold with constant sectional curvature equal to $k=\pm 1$ or $k=0$, and $f:I\to \mathbb{R}^{+}$ a smooth positive function;
\item[iv.] If $M^n_{k}$ is complete and simply connected, then $(V^{n+1},g)$ is globally isometric to a warped product $I\times_{f}M^n_{k}$, with $M^n_{k}$ equal to the unique simply connected space-form of constant sectional curvature $k$.\footnote{These space-forms are precisely the ones described in (\ref{SpaceForms})}
\end{enumerate}
\end{theo}
\begin{proof}
The first two items in the theorem are exactly those established in \cite[Chapter 12, Proposition 6]{ONeillBook}. To establish the third one, notice that given some fixed $s\in I$, for any $t\in I$ and any $p\in M$
\begin{align}
g_t(v,v)=h(t,s)g_s(v,v), \: \: \forall \: v\in T_pM,\: v\neq 0,
\end{align}
where $h(t,s)=g_t(\frac{v}{|v|_{g_s}},\frac{v}{|v|_{g_s}})>0$ is independent of the point $p\in M^n$ because of $(ii)$. For $s$ fixed, setting $\gamma\doteq g_s$ we find that
\begin{align}\label{Warping}
g=-dt^2+h\gamma,
\end{align}
where by $(i)$ $(M^n,\gamma)$ is a Riemannian manifold of constant sectional curvature. Notice that if the sectional curvature of $\gamma$ is equal to $C\neq 0$, then $\gamma_{\epsilon}\doteq |C|\gamma$ has sectional curvature equal to $\epsilon\doteq \mathrm{sign}(C)=\pm 1$. Thus, redefining the warping function in (\ref{Warping}) as $f^2(t)=|C|^{-1}h(t,s)>0$, we find 
\begin{align}
g=-dt^2+f^2(t)\gamma_{\epsilon}.
\end{align}
Finally, if $M^n$ is simply connected and $\gamma$ complete, then $(iv)$ follows from well-known characterisation of complete simply connected space-forms (see, for instance, \cite[Chapter 8, Theorem 4.1]{DoCarmo}).
\end{proof}

Let us note that the only step which we believe to be missing from the arguments of \cite{WaldBook} leading to the conclusions $(iii)-(iv)$ of the above theorem, is the proof of the second statement. That is, that the maps $\mu_{s,t}$ act by homotheties. We highlight that, as can be seen in \cite[Chapter 12, Proposition 6]{ONeillBook}, although this is not a particularly involved proof, it does not seem to be self-evident either, and that is why we give preference to the proof in this last reference as what we deem should be regarded as the standard one.

What should be ultimately highlighted from this section, is that Definition \ref{WaldIsotropy} does entail the well-known rigidity of cosmological space-times, captured by claims $(iii)-(iv)$ in Theorem \ref{ONthm}. Having established this and highlighted some of the subtleties which have not always been taken into account in standard literature, within Section \ref{DefsClassifications} we will present three \emph{other} well-known definitions for isotropy within cosmology, all of which appeal to the existence of symmetries which do not allow us to identify special directions in each tangent space to each rest space $M^n_t$. The way each of these definitions is introduced is subtly different, and therefore, due to the existence of the counterexamples to some traditional constructions provided by the BCCCM space-times described in the next section, we will analyse each such definition in detail in Section \ref{DefsClassifications}, and provide a classification which shall ultimately also show that the definition of space-isotropy provided in \cite{Weinberg} is equivalent to the one given in this section, and thus also entails the same consequences as those described by Theorem \ref{ONthm}.  

\subsection{The family BCCCM of cosmological space-times}\label{BCCCM}

In this section we shall introduce the family of space-times analysed in \cite{MiguelCosmologicalST}, and which we denote by BCCCM. For this, let us first consider the following ordinary differential equation problem for a real function defined on some interval $I\subset\mathbb{R}$:
\begin{align}\label{sn-functions}
\begin{split}
\begin{cases}
&f'' + kf=0,\\
&f(0)=0,\: f'(0)=1.
\end{cases}
\end{split}   
\end{align}
The above is nothing more that the Jacobi equation along a unit speed geodesic $\gamma$, in a space of constant sectional curvature $k$, for the components of a (non-trivial) Jacobi field orthogonal to $\gamma$ vanishing at a chosen origin and written in orthonormal frame along $\gamma$. We denote by $S_k:\mathbb{R}\to \mathbb{R}$ the unique solution corresponding to a fixed $k$, which implies
\begin{align}
\begin{split}
S_k(r)=\begin{cases}
\frac{\sin(\sqrt{k}r)}{\sqrt{k}}, &\text{ if } k>0,\\
r, &\text{ if } k=0,\\
\frac{\sinh(\sqrt{-k}r)}{\sqrt{-k}}, &\text{ if } k<0.
\end{cases}
\end{split}
\end{align} 
Using these function, the metric of the space-forms $M^n_k$ can be written in geodesic polar normal coordinates, centred at some chosen origin $p\in M^n_k$, as
\begin{align}
\gamma_k=dr^2+S_k^2(r)g_{\mathbb{S}^{n-1}},
\end{align}
where the coordinate $r$ above stands for the geodesic distance function from a chosen origin $p\in M^n_k$, and $g_{\mathbb{S}^{n-1}}$ stands for the canonical round metric on the $(n-1)$-dimensional sphere. 

Let us now consider a smooth function $k:I\subset \mathbb{R}\to \mathbb{R}$, the manifold $I\times \mathbb{R}^n\backslash\{0\}$, where we parametrise the first factor by a time-coordinate $t\in I$ and we see $\mathbb{R}^n\backslash\{0\}\cong \mathbb{R}^{+}\times S^{n-1}$, where $S^{n-1}$ stands for a topological sphere. We endow such manifold with the Lorentzian metric:
\begin{align}\label{MiguelMetric}
g_k=-dt^2+dr^2+S_{k(t)}^2(r)g_{\mathbb{S}^{n-1}}, \: \: (t,r)\in I\times \mathbb{R}^{+}.
\end{align}
Therefore, the above manifold is diffeomorphic to $P_{+}\times S^{n-1}$, where $P_{+}=\{(t,r)\in I\times \mathbb{R}^{+}\}$, and the Lorentzian manifold is then a warped product given by $V^{n+1}_{k}\doteq P_{+}\times_{S_{k(t)}}\mathbb{S}^{n-1}$, where we equip $P_{+}$ with the metric $g_{P}\doteq -dt^2+dr^2$. Let us notice that these metrics are smooth, since the functions:
\begin{align*}
S:I\times\mathbb{R}^{+}&\to \mathbb{R}^{+},\\
(t,r)&\mapsto S_{k(t)}(r)
\end{align*}
are smooth whenever $k:I\to \mathbb{R}$ is smooth. For a direct proof of this fact, see \cite[Lemma 2.1]{MiguelCosmologicalST}. Also notice that this follows from an application of general results on ODE theory, since the solution to (\ref{sn-functions}) is not only smooth with respect to the initial data, but also with respect to variations in the parameter $k$. Thus, being $S_{k(t)}(r)$ defined as the unique solution to (\ref{sn-functions}) for any fixed $k(t)$, we see that $k\in C^l(I)$, $l\geq 0$, implies that all the $l$-th partial derivatives $\partial^l_tS_{k(t)}(r)$ exist and are continuous, implying joint smoothness $S\in C^l(I\times \mathbb{R}^{+})$. As noted in \cite{MiguelCosmologicalST}, it is important that the smoothness of $S$ on $t$ is linked to the smoothness of $k$ itself and not $\sqrt{\pm k}$, since then one sees that, given a function such as $k(t)=-t^2$ ($\sqrt{-k(t)}=|t|$), the function $S_{k(t)}(r)$ is smooth at $t=0$, even though $t\mapsto |t|$ is not.

Let us highlight that, on the one hand, when the function $k$ satisfies $k(t)\leq 0$ for all $t\in I$, then $V^{n+1}_{k}$ extends smoothly to $r=0$, so $V^{n+1}_{k}\cong (I\times \mathbb{R}^n,g_k)$, and the slices $t=cte$ become Cauchy hypersurfaces due to \cite[Theorem 3.1]{MiguelCosmologicalST}. On the other hand, in regions where $k>0$, the metric (\ref{MiguelMetric}) on $I\times\mathbb{R}^n\backslash\{0\}$, extends smoothly both to $r=0$ and to $r_{\infty}=\frac{\pi}{\sqrt{k(t)}}$, making the slices $t=cte$ in such regions isometric to round spheres of intrinsic sectional curvature $k(t)$ \cite[Theorem 4.9]{MiguelCosmologicalST}.\footnote{In this positive curvature case, the second factor in $I\times\mathbb{R}^n$ is seen as the one point decompatification of a topological sphere $S^n$.} As shall be commented below, this second case becomes much more subtle in the case the function $k$ transition from $k>0$ to $k\leq 0$. Before entering into those details, let us introduce the following definition.

\begin{de}[BCCCM space-times]\label{BCCCM}
Given a smooth function $k:I\to\mathbb{R}$, $I\subset \mathbb{R}$ an open interval, we define the \emph{basic cosmological curvature change models} (BCCCM) as the manifolds $V^{n+1}=I\times \mathbb{R}^n\backslash\{0\}\cong P_{+}\times S^{n-1}$, equipped with a Lorentzian metric of the form (\ref{MiguelMetric}), extended naturally and smoothly to $r=0$, and, on any $t=t_0$ slice for which $k(t_0)>0$, also extended naturally to $r_{\infty}=\frac{\pi}{\sqrt{k(t_0)}}$.
\end{de}

We can present the following fundamental results associated to these BCCCM space-times:

\begin{theo}[BCCCM open models - Theorem 3.1 in \cite{MiguelCosmologicalST}]
Let $V^{n+1}_k$ be a BCCCM space-time for $k\leq 0$. Then (\ref{MiguelMetric}) is a smooth Lorentzian metric on the whole $\mathbb{R}^{n+1}=I\times\mathbb{R}^n$ with slices $t=t_0$ isometric to $\mathbb{E}^n$ if $k(t_0)=0$ and to $\mathbb{H}^n_{k(t_0)}$ when $k(t_0)<0$. Moreover, each slice $t = t_0$ is a Cauchy hypersurface.
\end{theo}

%When the function $k$ satisfies $k(t)\leq 0$ for all $t\in I$, then $V_{k}$ extends smoothly to $r=0$, so $V_{k}\cong (I\times \mathbb{R}^n,g_k)$, and the slices $t=cte$ become Cauchy hypersurfaces due to \cite[Theorem 3.1]{MiguelCosmologicalST}. In \cite{MiguelCosmologicalST}, these space-times are referred to as the \emph{open cosmological models of space curvature function} $t\mapsto k(t)$.\footnote{See \cite[Definition 3.3]{MiguelCosmologicalST}.} 

As was highlighted in \cite{MiguelCosmologicalST}, the cases where the spatial curvature function $k(t)$ changes sign and is positive somewhere work in the same way from the local viewpoint but are quite subtle from the global one. This can be partially understood by noticing that with the extension of the metric (\ref{MiguelMetric}) both to $r=0$ and to $r_{\infty}=\frac{\pi}{\sqrt{k(t)}}$ in such regions, one sees that in these positive curvature regions $V^{n+1}_{k}\cong I\times \mathbb{S}_{k(t)}^{n}$. Nevertheless, on any region where $k(t)\leq 0$ the extension to $r_{\infty}$ is clearly impossible, and the $t=cte$ hypersurfaces are topologically $\mathbb{R}^n$ and gometrically isometric to either $\mathbb{E}^n$ ($k=0$) or $\mathbb{H}_{k(t)}^n$ ($k(t)<0$). We therefore see that, if the function $k$ is somewhere positive and then changes sign, the \emph{intrinsically isotropic observers} defined by the global time-like vector field $\partial_t$ will experiment a topological change on their common simultaneity spaces. To highlight the subtleties this imposes, it implies that such $t=cte$ simultaneity spaces cannot all be Cauchy hypersurfaces \cite{GerochSplittingThm}. To deal with this case, in \cite{MiguelCosmologicalST} the author introduces a variation of the BCCCM space-times introduced in Definition \ref{BCCCM}, which are defined as \emph{basic cosmological topological curvature change models} (BCTCCM),\footnote{See \cite[Definition 4.8]{MiguelCosmologicalST}.} and are constructed to accommodate a curvature change from $k(t)>0$ for $t<0$ to $k(t)=0$ for $t\geq 0$. In particular they obey the following properties:

\begin{theo}[Theorem 4.9 in \cite{MiguelCosmologicalST}]
Any BCTCCM is a smooth spacetime satisfying:
\begin{enumerate}
\item All the slices $t = t_0$ have constant curvature isometric to the sphere of extrinsic radius $\frac{\pi^2}{k(t)}$ if $t_0 < 0$ and to $\mathbb{E}^n$ otherwise;
\item It is globally hyperbolic, with Cauchy hypersurfaces homeomorphic to $S^n$. In particular, the slices $t=t_0<0$ are Cauchy.
\end{enumerate}
\end{theo}

Although these BCTCCM are not exactly equal to the models we have described in Definition \ref{BCCCM}, they are isometric to such models in a neighbourhood of $t\geq 0$, with a curvature function $k:(-\epsilon,\infty)\to \mathbb{R}$, $\epsilon>0$, which changes sign at $t=0$ from positive to zero curvature.\footnote{For details, see equations (15)-(17) in \cite{MiguelCosmologicalST}, as well as equation (21) in the proof of Theorem 4.9 therein.} Let us furthermore notice that in \cite[Section 4.3]{MiguelCosmologicalST} it is explained how the above mentioned results for BCTCCM and the open model can be combined to produce a transition from positive to negative curvature, and thus accommodate curvature functions $k:I\to\mathbb{R}$ which change sign from $k<0$, say for $t<0$, to $k>0$, for $t>0$. These models end up being again globally hyperbolic with compact Cauchy surfaces, and are again isometric to (\ref{MiguelMetric}) in a neighbourhood of $t\geq 0$. 

With all of the above, we see that the family of BCCCM introduced in Definition \ref{BCCCM} provides us with a well-behaved (globally hyperbolic with complete Cauchy surfaces) family of cosmological space-times, where the simultaneity spaces associated with the observers defined by the integral curves of $\partial_t$ are \emph{intrinsically maximally symmetric}.\footnote{See Section \ref{DefsClassifications} for further details on what is meant by intrinsically maximally symmetric.} After classifying different notions for isotropy of cosmological space-times in Section \ref{DefsClassifications}, the \emph{non-trivial} BCCCM space-times will be shown to be non-isometric (even locally) to any FLRW space-time in Section \ref{RigiditySection1}.

\subsection{Some geometric conventions}

In order to analyse the curvature properties of the BCCCM space-times introduced in the previous section, we shall interpret our space-time as the warped product $P_{+}\times_{S_{k(t)}(r)}\mathbb{S}^{n-1}$, where $P_{+}$ denotes the half plane with coordinates $(t,r)$ in $\mathbb{R}^2$ defined by the condition $r>0$ and furnished with the flat metric $g_{P}\doteq -dt^2+dr^2$, while $\mathbb{S}^{n-1}$ denotes the round unit sphere. Let us also recall that, in this context, we refer to vector fields on the warped product which are tangent to the \emph{base} (in our case $P_{+}$) as \emph{horizontal}, and to those tangent to the \emph{fibre} (in our case $\mathbb{S}^{n-1}$) as \emph{vertical}.\footnote{For further details, see \cite[Chapter 7]{ONeillBook}.} Also, to avoid any ambiguity, let us make explicit the curvature conventions we follow in this text, where, given an $(n+1)$-dimensional semi-Riemannian manifold $(M^{n+1},g)$ and denoting by $\nabla$ its associated Riemannian connection, the curvature tensor is defined as:
\begin{align*}%\label{curvature1}
R(X,Y)Z=\nabla_X\nabla_YZ - \nabla_Y\nabla_XZ - \nabla_{[X,Y]}Z, \text{ for all } X,Y,Z\in \Gamma(TM).
\end{align*}
Also, given an arbitrary coordinate system $\{x^i\}$ on $M$, we label its components as follows:
%This convention is the opposite to O'Neill and to Besse, that is, $R_{ON}=-R$. Nevertheless, we also differ on how we label the components of this tensor. That is, in my notations, I define
\begin{align*}
\begin{split}
R^{i}_{jkl}&=dx^{i}(R(\partial_k,\partial_l)\partial_j)=\partial_k\Gamma^{i}_{lj}-\partial_l\Gamma^{i}_{kj} +  \Gamma^{i}_{ku}\Gamma^u_{jl}- \Gamma^{i}_{lu}\Gamma^u_{jk}.\\
%&={R_{{}_{ON}}}^i_{jkl}.
\end{split}
\end{align*}
%That is, when it comes to the components of the curvature tensor, we have the same convention, despite the fact that when we talk about the operator we use opposite conventions. This comes down to whether we relate the $\partial_l$ and $\partial_k$ to the third or fourth index of $R$. Then, concerning the Ricci tensor, we have that
From this we get the Ricci tensor from the following contraction:
\begin{align*}
\mathrm{Ric}_{ij}\doteq R^l_{ilj}.%= \mathrm{Ric}_{{ON}_{ij}}.
\end{align*}
We shall also consider the $(0,4)$-curvature tensor, given by\footnote{Due to the symmetries of the curvature tensor, our $(0,4)$-curvature tensor agrees with the one from \cite{Besse}, and thus our Weyl tensor also agrees with his. Therefore one can compare expressions with \cite[Definition 1.117]{Besse}.}
\begin{align*}
R(V,X,Y,Z)&=g(R(Y,Z)X,V).
%&=g(R(X,V)Y,Z)=-g(R(V,X)Y,Z)=g(R_{ON}(V,X)Y,Z)=R_{Besse}(V,X,Y,Z)
\end{align*}
%So, with these conventions, the $(0,4)$-curvature tensors agree for us and also for Besse-O'Neil. We can therefore compare the $(0,4)$-Weyl tensor with Besse.
With these conventions, one has that the Weyl tensor is given by
\begin{align}\label{WeylTensor}
\begin{split}
W(V,X,Y,Z)&=g(R(Y,Z)X,V) - \frac{1}{n-1}\mathrm{Ric}\KN g (V,X,Y,Z) + \frac{R_{g}}{2n(n-1)}g\KN g(V,X,Y,Z),
\end{split}
\end{align}
where $\KN:\Gamma(S_2M)\times \Gamma(S_2M)\mapsto \Gamma(T^0_4M)$, $S_2M$ the bundle of symmetric $(0,2)$-tensor fields, is the Kulkarni-Nomizu product defined by
\begin{align}
h\KN k(V,X,Y,Z)\doteq h(V,Y)k(X,Z) + h(X,Z)k(V,Y) - h(V,Z)k(X,Y) - h(X,Y)k(V,Z),
\end{align}
for all $h,k\in S_2M$.
%\begin{align}
%\begin{split}
%W(V,X,Y,Z)&=g(\bar{R}(Y,Z)X,V) \\
%&- \frac{1}{n-1}(\mathrm{Ric}(Y,V)g(X,Z) + \mathrm{Ric}(X,Z)g(V,Y) - \mathrm{Ric}(Z,V)g(Y,X) - \mathrm{Ric}(Y,X)g(Z,V)) \\
%&- \frac{R_{g}}{n(n-1)}(g(Z,V)g(Y,X) - g(Z,X)g(Y,V)),
%\end{split}
%\end{align}
Notice that (\ref{WeylTensor}) amounts to the following coordinate expression
\begin{align}
\begin{split}
W_{\rho\mu\lambda\nu}&=R_{\rho\mu\lambda\nu} - \frac{1}{n-1}(R_{\rho\lambda}g_{\mu\nu} + R_{\mu\nu}g_{\rho\lambda} - R_{\rho\nu}g_{\mu\lambda} - R_{\mu\lambda}g_{\rho\nu}) - \frac{R_{g}}{n(n-1)}(g_{\rho\nu}g_{\mu\lambda} - g_{\rho\lambda}g_{\mu\nu}).
\end{split}
\end{align}

In this context, let us also recall the following definition (see, for instance, \cite[Definition 1.164]{Besse}).
\begin{de}\label{ConformallFlatness}
A semi-Riemannian manifold $(V,g)$ is said to be conformally flat if, for any $x\in V$, there exists a neighbourhood $\mathcal{U}$ of $x$ and a (smooth) function $f$ on $\mathcal{U}$, such that $(\mathcal{U},e^{2f}g)$ is flat.
\end{de}

Finally, let us recall the following well-known result:\footnote{See, for instance, \cite[Theorem 1.165]{Besse} and \cite{Eis} for more details.}
\begin{theo}\label{ConformallFlatnessTHM}
An $n$-dimensional semi-Riemannian manifold $(M^n,g)$, $n\geq 4$, is conformally flat if and only if its Weyl tensor vanishes.
\end{theo}

%\section{Rigidity and non-rigidity of cosmological space-times}

\section{A classification of isotropy definitions}\label{DefsClassifications}

In this section we shall start contrasting Definition \ref{WaldIsotropy}, which we have taken as the default characterisation of isotropic cosmological space-times, with other definitions well-known from the literature, all of which have been claimed to lead to the rigidity properties $(iii)$-$(iv)$ of Theorem \ref{ONthm}. We shall start this discussion introducing two SI definitions, which appeal to the existence of symmetries on each $M^n_t$ which capture the essence of \emph{space-isotropy}.

\begin{de}[SI cosmological space-times - Type I]\label{CBIsotropy}
A cosmological space-time is said to be isotropic if for every point $q\in V^{n+1}$, $q=(t,p)$ with $p=\pi(q)$, the (intrinsic) sectional curvature $k_{p}(P)$ of $(M^n_t,g_t)$ at $p$ is independent of the plane $P\subset T_pM^n_t$.
\end{de}

The above definition is actually extracted from the discussion around equation (2.1) in \cite[Chapter V, Section 2]{CB-Book}. More precisely, we quote:\footnote{See \cite[Page 107]{CB-Book}.}${}^{,}$\footnote{In the quotation below, the superscripts $3$ and $4$ are used to indicate that the corresponding metric is defined on a manifold of the corresponding dimensionality. Just as we do when writing the metrics $g_t$ in Definition \ref{CosmoligcalSTDef}, there is a slight abuse of notation for the ${}^{3}\!g$ parts.}
\begin{quote}
A cosmos satisfying the cosmological principle is a Lorentzian manifold $(\mathbb{R}\times M,{}^{4}\!g)$ with a metric of type ${}^{4}\!g=-dt^2+{}^{3}\!g$, such that, for each $t$, the Riemannian manifold $(M,{}^{3}\!g)$, that is \emph{the universe}, \emph{is isotropic and homogeneous}.
\end{quote}
Within Definition \ref{CBIsotropy} we are neglecting the reference to homogeneity since it is not relevant for our discussion.\footnote{In contrast to the definition of isotropy, the definition of homogeneity is (more) uniform through standard literature. In particular, a Riemannian manifold is said to be (locally) homogeneous if it admits a transitive group of (local) isometries.} Putting together the above quote with the definition of isotropy of a Riemannian manifold given in \cite[Chapter V, Definition 3.1]{CB-Book}, one obtains Definition \ref{CBIsotropy} above, and also one finds that the classic Schur's Lemma implies that if $(M,{}^{3}\!g)$ is isotropic around every point, then it is a space of constant sectional curvature and thus (locally) homogeneous.\footnote{See, for instance, \cite[Chapter V, Theorem 3.4 and Theorem 3.7]{CB-Book}.} The additional hypothesis of simply connectedness would also imply that isotropy at every point implies global homogeneity.

Although a priori somewhat different, we also classify the definition of isotropy provided in \cite{CarrollBook} within the SI category. In this case, the author presents the following \emph{definition} for \emph{isotropic cosmological space-times}:\footnote{See \cite[Page 329]{CarrollBook}.}
\begin{quote}
To describe the real world, we are forced to give up the ``perfect" Copernican principle, which implies symmetry throughout space and time, and postulate something more forgiving. It turns out to be straightforward, and consistent with observation, to posit that the universe is spatially homogeneous and isotropic, but evolving in time. In general relativity this translates into the statement that the 
universe can be foliated into spacelike slices such that each three-dimensional slice is maximally symmetric. We therefore consider our spacetime to be $\mathbb{R}\times \Sigma$, where $\Sigma$ is maximally symmetric.
\end{quote}

One should put together the above quote with the definition of \emph{maximally symmetric space}, which in this reference is given in \cite[Chapter 3, Section 3.9]{CarrollBook}.\footnote{One should also put together the above quotation from \cite{CarrollBook} with the definition of isotropy provided in the same reference in Chapter 8, Section 8.1, page 323. Putting all these things together, it seems clear to us that the intended definition of isotropy makes reference to the existence of isometries of $(M_t,g_t)$ for each $t$, in contrast to Definition \ref{WaldIsotropy}.} In this context, a $d$-dimensional (semi-)Riemannian manifold $(M^d,g)$ is said to be maximally symmetric if it admits in a neighbourhood of any point the \emph{maximum number of linearly independent Killing vector fields that is possible for a $d$-dimensional space.} For a rigorous proof that such maximum number is $N_{max}(d)=\frac{d(d+1)}{2}$ see \cite[Chapter 9, Lemma 28]{ONeillBook}.\footnote{In \cite{Weinberg} this result is used, although it seems to be established only for analytic metrics and their corresponding analytic Killing vector fields.} In this case, disentangling homogeneity from isotropy is not as direct as in the previous one. Nevertheless, in \cite[Chapter 3, Section 3.9]{CarrollBook}, the author is clear in specifying that this maximum number of local Killing fields consist on a family of $d$-\emph{translational} Killing fields plus a family of $\frac{d(d-1)}{2}$ \emph{rotational} Killing fields, where the latter produce isometries fixing a given point $p\in M^n$ and rotating any given unit vector $v\in T_pM$ into another chosen unit vector $w\in T_pM$. Once again, putting this together with the discussion in \cite[Chapter 8, Section 8.1, page 323]{CarrollBook}, it seems clear to us that this second family of $\frac{d(d-1)}{2}$ \emph{rotational} Killing fields is what captures the intended contribution of isotropy within the discussion of FLRW space-times in this reference.\footnote{Let us also draw the reader's attention to a similar discussion presented in the classic textbook \cite[Chapter 13]{Weinberg}. Part of our presentation below overlaps with the discussion in this reference, although through quite different techniques and lines of argument. In particular, our intention is to present clear-cut geometric results, with explicit hypotheses, which validate certain intuitions around how these so-called rotational Killing fields do represent rotations on a general setting. Our presentation is tailored for our purposes and we believe contributes to bridging small gaps in the existing literature.} To make all this more mathematically clear, let us highlight the following result, whose proof can be consulted in \cite[Chapter 9, Lemma 28]{ONeillBook}.
\begin{lem}\label{MaximumKilling}
Let $(M^d,g)$ be a semi-Riemannian manifold. We denote by $\iota(M)$ the Lie algebra of Killing vector fields on $(M^d,g)$. Then, given $p\in M$, the linear map $E:\iota(M)\to T_pM\times\mathfrak{o}(T_pM)$, given by
\begin{align}
E(X)=(X_p,\nabla X_p)
\end{align}
is injective, where $\mathfrak{o}(T_pM)$ denotes the set of antisymmetric endomorphisms of $T_pM$. In particular $\mathrm{dim}(\iota(M))\leq N_{max}(d)$.
\end{lem}
%Here, $\mathfrak{o}(T_pM)$ denotes the set of antisymmetric endomorphisms of $T_pM$, $D$ denotes the appropriate Riemannian connection associated with $g$, and $DX_p$ is seen as the operator
%\begin{align*}
%DX_p:T_pM&\mapsto T_pM\\
%v &\mapsto D_vX\big\vert_p.
%\end{align*}
We shall denote by $\iota_p(M)\doteq\{X\in \iota(M) \: :\: X_p=0 \text{ for a given } p\in M\}$. In this case, the restriction
\begin{align*}
E_p\doteq E|_{\iota_p(M)}:\iota_p(M)\to \mathfrak{o}(T_pM)
\end{align*} 
is a linear injective map by the above lemma, such that, for each element $X\in \iota_p(M)$, $E_p(X)=\nabla X_p\in \mathfrak{o}(T_pM)$. To make perfect sense of the notion of isotropy described above from the quote of \cite{CarrollBook} disentangled from homogeneity, one would like to prove that, given a Riemannian manifold $(M^n,g)$, if $\mathrm{dim}(\iota_p(M^n))=\frac{n(n-1)}{2}$, then the collection of the corresponding 1-parameter groups of isometries $\varphi^X_s$, $X\in \iota_p(M)$, act as arbitrary point wise rotations on $T_pM$. That is, given an arbitrary element $A\in SO(T_pM)$, there is a Killing field $X\in \iota_p(M)$ such that $(d\varphi^X_s)_p= A$. That this holds is the content of the following result:
\begin{lem}\label{SIuniqueness}
Let $(M^n,g)$ be a Riemannian manifold, $p\in M$, $\mathcal{U}_p$ a neighbourhood of $p$ and assume that $\mathrm{dim}(\iota_p(\mathcal{U}_p))=\frac{n(n-1)}{2}$. Let $\varphi^X_s$ denote the associated flow of $X\in \iota_p(\mathcal{U}_p)$. Then, the set 
\begin{align}\label{SO-Surjectivity}
\{(d\varphi^X_s)_p\: :\: X\in \iota(\mathcal{U}_p) \}=SO(T_pM).
\end{align}
In particular, $E_p$ is an isomorphism.
\end{lem}
\begin{proof}
%The result follows from the proof of the second half of Theorem \ref{Equivalence}. More explicitly, substituting the isometries $\phi_s^X\circ i$ by $\varphi^X_s$ as defined in this lemma
Since $X_p=0$, by uniqueness of integral curves, we see that $\varphi^X_s(p)=p\: \forall \: s$, and thus $(d\varphi^X_s)_p:T_pM\to T_pM$ is a 1-parameter group of linear isometries of $(T_pM,g_p)$. Picking an orthonormal basis $\{e_i\}_{i=1}^n$ for $T_pM$ to represent $(d\varphi^X_s)_p$, we see that the matrix $((d\varphi^X_s)_p)^j_i=g_p(e_j,(d\varphi^X_s)_p(e_i))$ stands for a continuous 1-parameter group of linear isometries of Euclidean $n$-dimensional space which belongs to the connected component of the identity. That is, these are elements of $SO(n)$. We still need to see that every element in $SO(n)$ can be realised in this manner, but for that one can first recall that any smooth one parameter subgroup $\gamma (s)$ of a Lie group $G$ must be of the from $\exp(s\gamma'(0))$, where $\exp:\mathfrak{g}\to G$ denotes the exponential map associated to the Lie group $G$ and $\mathfrak{g}$ its Lie algebra. This therefore implies that $(d\varphi^X_s)_p=\exp(s\gamma'(0))$, for some $\gamma'(0)\in \mathfrak{o}(T_pM)$. Also, using Proposition \ref{RotationalKillingsProp} established below, we see that 
\begin{align*}
%(d\varphi^X_s)_p(v)&=d\pi_{q}\circ (d\varphi^X_s)_q(i(v)),\\
\frac{d(d\varphi^X_s)_p(v)}{ds}\Big\vert_{s=0}&=D_vX \text{ for any } v\in T_pM,
\end{align*}
where in the last equality $D:\Gamma(TM)\times\Gamma(TM)\to \Gamma(TM)$ stands for the Riemannian connection on $(M^n,g)$. Thus, referring to the linear map on $E_p:T_pM\to T_pM$ which acts by $v\mapsto D_vX$, we find
\begin{align}
\frac{d(d\varphi^X_s)_p}{ds}\Big\vert_{s=0}&=E_p(X)=\gamma'(0),
\end{align}
and hence $(d\varphi^X_s)_p=\exp(sE_p(X))$.
%Denoting $\psi^X_s\doteq (d\phi^{X}_s)_{q}\circ i$ and $w_s\doteq \psi_s^X(w)$ for any fixed $w\in T_pM$, we see that  (\ref{IsotropyEquiv.0})now gives us that\footnote{Notice that (\ref{IsotropyEquiv.0}) remains valid, since for its validity we only appealed to the property $X_q=0$.}
%\begin{align}
%\frac{dw_s^{i}}{ds}=\nabla_jX^i\big\vert_q w^j.
%\end{align}
%Denoting by $\sigma^i_j(X)\doteq \nabla_jX^i\vert_q\in \mathfrak{o}(n)$, the unique solution to the above system with initial data $w_0=w$ is given by $w_s=\exp(s\sigma(X))w$, which implies that $\psi_s^{X}=\exp(s\sigma(X)):T_pM\mapsto T_pM$.
By hypothesis, there are $\frac{n(n-1)}{2}=\mathrm{dim}(\mathfrak{o}(T_pM))$ linearly independent Killing vectors in $\iota_p(M)$, which implies that the map $E_{p}:\iota_p(M)\to \mathfrak{o}(T_pM)$ defined in Lemma \ref{MaximumKilling} is an isomorphism. Thus, given $\sigma_p\in \mathfrak{o}(T_pM)$, there is some $X\in \iota_p(M)$ such that $E_p(X)=\sigma_p$, and the set of isometries of $(T_pM,g|_{T_pM})$ given by
\begin{align*}
\{(d\varphi^X_s)_p \: : \: X\in \iota(\mathcal{U}_p)\}&=\{\exp(sE_p(X)):T_pM\to T_pM \: : \: X \in \iota_p(M)\}=\exp(\mathfrak{o}(T_pM)).
\end{align*}
Since $\exp:\mathfrak{o}(n)\to SO(n)$ is surjective, then (\ref{SO-Surjectivity}) follows.% this finally implies that, given any pair of vectors $v,w\in T_pM$ there is an \emph{isotropy local Killing field} $X$ such that the associated local isometry $\varphi_s^{X}=\mathrm{Id}\times \phi_s^{X}$ fixes $q$ and rotates $v,w$.
\end{proof}

\begin{prop}\label{RotationalKillingsProp}
Let $(M^d,g)$ be a semi-Riemannian manifold. Fix a point $q\in M$, let $X\in \iota_q(M)$ and denote by $\varphi_{\theta}$ the flow of $X$. Then, the linear map $(d\varphi_{\theta})_q:T_qM\to T_qM$ obeys the following identity:
\begin{align}
\frac{d(d\varphi_{\theta})_q(Y)}{d\theta}\Big\vert_{\theta=0}=\nabla_{Y_q}X, \: \forall \: Y\in \Gamma(T\mathcal{U}_q),
\end{align}
where $\mathcal{U}_q$ is an arbitrary neighbourhood of $q$.
\end{prop}
\begin{proof}
%First, notice that any vector field $Y\in \Gamma(\mathcal{U}_q)$ can be extended to a vector field $\tilde{Y}\in \Gamma(M)$ such that $Y=\tilde{Y}$ in a neighbourhood of $q$ just by multiplying by an appropriate cut-off function which equals one in a neighbourhood of $q$. Thus, we shall consider that such procedure has been done, and being concerned with a local result, we shall not distinguish $Y$ from $\tilde{Y}$ notationally. So, from now on we shall consider vector fields defined on all of $M$.
Since we are concerned with a local result, we assume that $\mathcal{U}_q$ is small enough so that $X\in \iota_q(M)$ is non-zero on $\mathcal{U}_q\backslash\{q\}$ and $\varphi_{\theta}\vert_{\mathcal{U}_q}$ is a diffeomorphism onto its image. Consider then the pull-back $\varphi_{\theta}^{*}$ acting on vector fields $Y\in \Gamma(T\mathcal{U}_q)$ as 
%the (smooth) 1-parameter groups of isometries $\phi_{\theta}$ of $\mathbb{S}^{n-1}$ give rise to (smooth) 1-parameter groups of space-time isometries of $g=-dt^2+f^2(t)g_k$ in a neighbourhood of $q=(t,p)$, given by $\varphi_{\theta}:I\times(0,\epsilon)\times \mathbb{S}^{n-1}$, $\varphi_{\theta}(t,r,s)=(t,s,\phi_{\theta}(s))$. In this way we distinguish the 1-parameter groups of rotations $\mathcal{R}_{\theta}$ on $T_pM$ around the origin in this system of coordinates, and denote by $\varphi_{\theta}=Id\times\phi_{\theta}$ the associated space-time isometry satisfying $(d\phi_{\theta})_p=\mathcal{R}_{\theta}$.
%Denoting by $\varphi_{\theta}$ one such rotation of angle $\theta$ and by $X_{\theta}$ be the associated (local) Killing field, since $\varphi_{\theta}(q)=q$ $\forall \theta$, then $X_{\theta}(q)=\frac{d\varphi_{\theta}}{d\theta}(q)=0$. Now, to check the second condition in Definition \ref{WeinbergIsotropy}, consider the pull-back $\varphi_{\theta}^{*}$ acting on vector fields $Y\in \mathfrak{X}(V)$ as 
\begin{align*}
%\varphi^{*}_{\theta}:T_rV\mapsto T_{\varphi^{-1}_{\theta}(r)}V
(\varphi^{*}_{\theta}(Y))_{\varphi^{-1}_{\theta}(m)}=(d\varphi^{-1}_{\theta})_m(Y_m), \: \forall \: m\in \mathcal{U}_q,
\end{align*}
%since $\varphi_{\theta}$ is smooth jointly on $(\theta,q)$, working in a normal coordinate systems $\{x^{\alpha}\}_{\alpha=0}^{n}$ around $q$, constructed from an orthonormal frame at $q$ with $e_0=\partial_t$, we can compute
In particular, since $\varphi^{-1}_{\theta}(q)=q$, then $\varphi^{*}_{\theta}\vert_q:T_qM\to T_qM$. Also, using that $\varphi^{-1}_{\theta}=\varphi_{-\theta}$ and that $d\varphi^{-1}_{\theta}$ is smooth jointly on $(\theta,m)\in (-\delta,\delta)\times \mathcal{U}_q$,\footnote{See, for instance, \cite[Chapter 4, Theorem 4.1.5]{AMR}.} it follows from \cite[Chapter 6, Theorem 6.4.1]{AMR} that
\begin{align*}
\frac{d(d\varphi^{-1}_{-\theta})_q(Y_q)}{d\theta}&=\frac{d\varphi^{*}_{-\theta}(Y)}{d\theta}\big\vert_{q}=d\varphi^{-1}_{-\theta}(\pounds_{-X_{\varphi_{-\theta}(q)}}Y)\big\vert_{q}=-[d\varphi^{-1}_{-\theta}(X),d\varphi^{-1}_{-\theta}(Y)]\vert_{q}=\nabla_{(d\varphi^{-1}_{-\theta})_q(Y_q)}d\varphi^{-1}_{-\theta}(X).
\end{align*}
Then, since $d\varphi^{-1}_{-\theta}(Y)=d\varphi_{\theta}(Y)$, we find that
\begin{align}\label{IsotropyEquiv.00}
\frac{d(d\varphi_{\theta})_q(Y_q)}{d\theta}=\nabla_{(d\varphi_{\theta})_q(Y_q)}d\varphi_{\theta}(X) \:\:\: \forall \: Y\in \Gamma(T\mathcal{U}_q).
\end{align}
We can simplify the above as follows. Given any point $m$ in the domain of $\varphi_{\theta}$, we may write $X_m=\frac{d\varphi_s}{ds}(0)$ where $\alpha(s)=\varphi_s(m)$ stands for the integral curve of $X$ starting at $\alpha(0)=m$. Then
\begin{align*}
(d\varphi_{\theta})_m(X_m)=(d\varphi_{\theta})_m\left(\frac{d\varphi_s}{ds}(0)\right)=\frac{d(\varphi_{\theta}\circ\varphi_s)}{ds}\vert_{s=0}=\frac{d\varphi_{\theta+s}}{ds}\vert_{s=0}=X_{\varphi_{\theta}(m)}.
\end{align*}
Therefore (\ref{IsotropyEquiv.00}) is equivalent to
\begin{align}\label{IsotropyEquiv.0}
\frac{d(d\varphi_{\theta})_q(Y_q)}{d\theta}=\nabla_{(d\varphi_{\theta})_q(Y_q)}X \:\:\: \forall \: Y\in \Gamma(T\mathcal{U}_q).
\end{align}
For $\theta=0$, the above gives us
\begin{align}\label{IsotropyEquiv.01}
\frac{d(d\varphi_{\theta})_q(Y_q)}{d\theta}\Big\vert_{\theta=0}=\nabla_{(d\varphi_{0})_q(Y_q)}X=\nabla_{Y_q}X
\end{align}
\end{proof}

Lemma \ref{SIuniqueness} above, and in particular its proof, justify why Killing fields which fix a given point can actually be seen as generating point wise rotations. A similar (and actually easier) result also justifies why a (non-trivial) Killing field satisfying $\nabla X\vert_p=0$ can be interpreted as generating a translational symmetry. With all this information, we can now write down what we consider is the definition of isotropy clearly intended in \cite{CarrollBook} in a way which is more useful for our purposes.

\begin{de}[SI cosmological space-times - Type II]\label{CarrollIsotropyDef}
A cosmological space-time $(V^{n+1},g)$ is said to be isotropic if, for each $t$, the Riemannian manifolds $(M^n_t,g_t)$ representing the rest-spaces of the observers defined by $U=\partial_t$, in a neighbourhood of any point $p\in M^n_t$, admit a maximal family of rotational Killing fields. That is, if given $p\in M^n_t$, there is a neighbourhood $\mathcal{V}_p\subset M^n_t$, such that $\mathrm{dim}(\iota_p(\mathcal{V}_p))=\frac{n(n-1)}{2}$. 
\end{de}

As we have previously stated, we classify both Definition \ref{CBIsotropy} and Definition \ref{CarrollIsotropyDef} as SI-definitions because they make reference to isometries of the rest spaces associated with the observers, and not to space-time isometries. In particular, these two definitions are equivalent, as shown below.

\begin{theo}\label{SIuniquenessTHM}
A cosmological space-time $(V^{n+1},g)$ is SI type I iff it is SI type II.
\end{theo}
\begin{proof}
From Definition \ref{CBIsotropy}, if $(V^{n+1},g)$ is SI type I, then $(M_t^n,g_t)$ must have constant sectional curvature, so each $(M_t^n,g_t)$ is locally isometric to a simply connected space-form, all of which are (locally) maximally symmetric.\footnote{An explicit proof could also be written through a minor adaptation of the first part of the proof of Theorem \ref{STIuniquenessTHM} below.}

To see the converse, consider some fixed $(M^n_t,g_t)$, $p\in M^n_t$ and two planes $P_1,P_2\subset T_pM^n_t$. Let $\{v_1,v_2\}$ and $\{w_1,w_2\}$ be two orthonormal bases generating $P_1$ and $P_2$ respectively. Let us complete these bases as orthonormal bases for $T_pM^n_t$ with the same orientation, given by $B_1=\{v_1,v_2,v_3,\cdots,v_n\}$ and $B_2=\{w_1,w_2,w_3,\cdots,w_n\}$. Then, there is matrix $A\in SO(T_pM^n_t)$ associated to the change of basis $B_1$ to $B_2$ mapping $v_i\mapsto \omega_i$ for $i=1,2$. If $(V^{n+1},g)$ is SI type II, there is a neighbourhood $\mathcal{V}_p\subset M^n_t$ such that $\mathrm{dim}(\iota_p(\mathcal{V}_p))=\frac{n(n-1)}{2}$. Then Lemma \ref{SIuniqueness} implies there is an isometry $\varphi$ of $(\mathcal{V}_p,g_t)$ such that $d\varphi_p=A$, so in particular $(d\varphi)_p(v_i)=w_i$, $i=1,2$. This in turn implies that the sectional curvatures of these planes are equal. That is, $k_p(P_1)=k_p(P_2)$, and therefore the claim follows.
\end{proof}

\begin{remark}\label{NonRigRemark1}
Consider a FLRW space-time, so that $V^{n+1}\cong I\times_{f}M^n_k$ satisfies claim (iv) in Theorem \ref{ONthm}. In this case, it follows directly that the slices $M^n_t$, which are homothetic to $M^n_k$, are spaces of constant sectional curvature and thus satisfy Definition \ref{CBIsotropy} and hence also Definition \ref{CarrollIsotropyDef}. This last claim is also obvious since these simultaneity spaces are known to be (intrinsically) maximally symmetric. So clearly, the FLRW family is contained within the SI family of space-times. The usual claim is that there is a (local) converse to this statement. Nevertheless, notice that any BCCCM model described in Section \ref{BCCCM} satisfies the SI criteria for isotropy. For instance, using Definition \ref{CBIsotropy}, this is obvious since the only thing to be checked is the independence of $k_p(P)$ on $P\subset T_pM^n_t$, which in these cases follows from the fact that $M^n_t\cong M^n_{k(t)}$ and have therefore constant sectional curvature. This by itself already implies that a space-time can be constructed satisfying the SI criteria for isotropy and without looking like a FLRW space-time. 

The above should provide some caution to classical claims which start with SI definitions, showing that even after knowing that each slice $M^n_t$ associated with the isotropic observers is a space of constant sectional curvature, one has still some work to do to prove that the space-time splits (even locally) as a warped product $I\times_fM^n_k$. Nevertheless, in the next section we will that actually see any non-trivial BCCCM space-time is not even locally isometric to any FLRW, so that there is no hope in reconciling SI definitions of isotropy with rigidity with the FLRW family. 
\end{remark}

Let us now differentiate the above (equivalent) SI definitions of isotropy from the one used in \cite[Chapter 13, Section 5]{Weinberg}. This definition is quite similar to the one in \cite{CarrollBook}, but it takes into consideration that the Killing fields which make the space-slices maximally symmetric, must be induced by space-time Killing fields. That is, given a point $q\in V^{n+1}$ and a neighbourhood $\mathcal{U}_q\subset V^{n+1}$, Killing vector fields in this family defined on $\mathcal{U}_q$ are supposed to be tangent to the space-like symmetric hypersurfaces, say $M^n_t\cap \mathcal{U}_q$ with $t\in I$, and are actually taken to be \emph{maximal}, in the sense that $(M^n_t\cap \mathcal{U}_q,g_t)$ is then maximally symmetric. Once again, this last condition will imply also the existence of local translational symmetries, and therefore (from the beginning) one assumes local homogeneity. Being only concerned with isotropy, from the above discussion we now know how to distinguish the subset of $\iota(M_t)$ which represents point wise rotations and we can therefore preserve the intended definition of isotropy in \cite{Weinberg} by demanding the following:%\footnote{This is actually extracted from the definition of isotropic space in \cite[page 379]{Weinberg} put together with the way cosmological space-times are analysed in Section 5 of Chapter 13, specially equations (13.5.1)-(13.5.3) and the paragraph after equation (13.5.2).}% One can do without this last requirement as long as one imposes a characterisation of isotropy to each $(M_t,g_t)$ that implies constant sectional curvature, which, in turn, implies the local existence of such translational isometries. Such a characterisation can be done in the spirit of \cite{Weinberg} by paying attention at how the definitions are introduced and used. Doing this, one can capture the intended definition of isotropy as follows:\footnote{This is actually extracted from the definition of isotropic space in \cite[page 379]{Weinberg} put together with the way cosmological space-times are analysed in Section 5 of Chapter 13, specially equations (13.5.1)-(13.5.3) and the paragraph after equation (13.5.2).}

\begin{de}[STI cosmological space-time - Type II]\label{WeinbergIsotropy}
A cosmological space-time is said to be isotropic at a point $q=(t,p)\in V^{n+1}$, $p\in M^n$, if there exists a neighbourhood $\mathcal{U}_q\subset V^{n+1}$ and a subset $\overline{\iota_q(\mathcal{U}_q)}$ of $\iota_q(\mathcal{U}_q)$ % vector fields $\{X^a\}_{a=1}^{N\leq N_{max}(n)}$, which for each $(s,m)\in V$ in their domain are tangent to $M$ and 
such that %each element $X\in \overline{\iota_0(\mathcal{U}_q)}$ satisfies
\begin{enumerate}
\item $\overline{\iota_q(\mathcal{U}_q)}=\{X\in \iota_q(\mathcal{U}_q) \: : \: X=X^{\top} \}$, where $X^{\top}=d\pi(X)$. That is, $\overline{\iota_q(\mathcal{U}_q)}$ consists of elements of $\iota_q(\mathcal{U}_q)$ which are tangent to each $M^n_t$; %$X^a_{q}=0$ for all $a=1,\cdots,N$;
\item The map $\overline{E}^{\top}_q:\overline{\iota_q(\mathcal{U}_q)}\to \mathfrak{o}(T_pM^n_t)$, given by 
\begin{align}
\begin{split}
\overline{E}^{\top}_q(X)=(\nabla X)^{\top}:T_{\pi(q)}M^n_t&\to T_{\pi(q)}M^n_t,\\
v&\mapsto  (\nabla_{\bar{v}} X)^{\top}
\end{split}
\end{align}
is surjective, where again $Y^{\top}=d\pi(Y)$ denotes the projection of $Y\in \Gamma(T\mathcal{U}_q)$ tangential to each $M_t$ and $\bar{v}$ denotes any extension of $v$ to a vector field on $\mathcal{U}_q$. %The antisymmetric matrices $\{\nabla_{i}X^a_{j}|_{q}, \: i,j=1,\cdots,n \: ; \: a\in\mathcal{I}\}$ span the space of anti-symmetric matrices, where $\nabla$ stands for the space-time Riemannian connection and $X^a_{j}$ denote the components of the 1-form metrically equivalent to $X^a$ in any coordinate system $\{x^0,x^i\}^n_{i=1}$ adapted to $M$.\footnote{That is, locally $M$ is given by the condition $x^0=0$.}
\end{enumerate}

Finally, we say that a cosmological space-time is isotropic if it is isotropic at every point $q=(t,p)\in V^{n+1}$.
\end{de}

\begin{remark}
In \cite{Weinberg} there is no actual explicit definition written as in Definition \ref{WeinbergIsotropy} above. There, the author actually provides the following definition of isotropy in page 378:
\begin{quote}
A metric space is said to be isotropic about a given point $X$ if there exist infinitesimal isometries (13.1.3) that leave the point $X$ fixed, so that $\xi^{\lambda}(X) = 0$, and for which the first derivatives $\xi_{\lambda;\nu}$ take all possible values, subject only to the antisymmetric condition (13.1.5).
\end{quote}
In the above quote, the referred equation (13.1.3) is used to define the the cited isometries as generated by the Killing fields $\xi$. Then, this Killing field is represented by its components $\xi^{\lambda}$ in an arbitrary coordinate system, and the following notation is used for covariant differentiation $\xi_{\lambda;\nu}\doteq \nabla_{\nu}\xi_{\lambda}$. Nevertheless, when studying cosmological isotropic space-times, the author analyses spaces with maximally symmetric subspaces, and not space-times which are isotropic under the above quoted definition. This in done in \cite[Section 5 of Chapter 13]{Weinberg}, where in the case of interest, in which the maximally symmetric subspace is a hypersurface $M^n$, the author considers the existence of a family of space-time Killing vector fields in a neighbourhood of a point $q\in V^{n+1}$, which induce on $M^n_t$ a family of $\frac{n(n+1)}{2}$ independent Killing fields. Looking at equations (13.5.1)-(13.5.3) we see that these Killing fields are taken to be tangent to each $M^n_t$ (our condition 1 in Definition \ref{WeinbergIsotropy}), and also that within the $\frac{n(n+1)}{2}$-dimensional family, there is a subfamily such that 
\begin{align}\label{WeinbergRotCondition}
\xi_{k;l}=g(\partial_k,\nabla_l\xi)
\end{align}
can take arbitrary antisymmetric values at $q$, where $\{x^0,x^k\}_{k=1}^n$ is taken to be a system of coordinates for $V^{n+1}$ around $q$ which is adapted to $M^n_t$.\footnote{That such system of coordinates is adapted to $M^n$ means that $\{x^{k}\}_{k=1}^n$ is a system of coordinates for $M^n$ around $p=\pi(q)$.} From our discussion associated to Lemma \ref{SIuniqueness}, we distinguish this subfamily as the one which is actually associated to rotational symmetry and thus to isotropy, which can be disentangled from the $n$-additional Killing fields intended to generate translational isometries associated with homogeneity. Let us then notice that (\ref{WeinbergRotCondition}) is precisely our condition 2 in Definition \ref{WeinbergIsotropy}. The attentive reader can recognise that Definition \ref{WeinbergIsotropy} is the definition actually used in \cite[Section 5 of Chapter 13]{Weinberg} in order to pursue the claim that the associated cosmological space-times must be in the FLRW family.
\end{remark}

%In the above line of argument there seems to be at least one element missing in most standard references. For instance, if one follows \cite[Chapter 5]{WaldBook} or \cite[Chapter 5, Sections 3 and 4]{CB-Book}, after one establishes that $g_t$ must be a metric of constant sectional curvature for each $t$, it is not absolutely obvious that $g_t=S_k(t)\gamma_{\epsilon}$. This is in fact true under Definition \ref{WaldIsotropy},  but it is not necessarily true under Definition \ref{CBIsotropy}. The positive result follows from \cite[Chapter 12, Proposition 6]{ONeillBook} and the discussion presented there after the cited proposition. For the sake of completeness, let us state the corresponding results in a manner which is useful for our discussion.
Noticing that the Killing fields in the above definition will generate local space-time symmetries, which act on $U^{\perp}$ by construction, already makes contact between Definition \ref{WeinbergIsotropy} and Definition \ref{WaldIsotropy}. Along the same lines of Theorem \ref{SIuniquenessTHM}, we would like to prove that Definition \ref{WeinbergIsotropy} is equivalent to our default Definition \ref{WaldIsotropy}. Before doing this, we once more highlight that, if one were to be convinced that both these STI definitions lead to rigidity withing the FLRW family, then equivalence would follow immediately, since any FLRW space-time satisfies both of them (and clearly the SI-definitions as well). Nevertheless, because of the subtleties that we have exposed and the non-rigidity that we have commented for the SI-definitions (which shall be ultimately proven in Theorem \ref{BCCCMConfFlatThm}), we take only the rigidity associated to Definition \ref{WaldIsotropy} as given by Theorem \ref{ONthm} for granted.

\begin{theo}\label{STIuniquenessTHM}
A cosmological space-time $(V^{n+1},g)$ is STI of type I if and only if it is STI of type II.
\end{theo}
\begin{proof}

Let us first prove that Definition \ref{WaldIsotropy} implies Definition \ref{WeinbergIsotropy}. With this in mind, let us assume isotropy under Definition \ref{WaldIsotropy}, consider an arbitrary point $q\in V^{n+1}$, and prove that $V^{n+1}$ is isotropic at $q$ under Definition \ref{WeinbergIsotropy}. Being such a point arbitrary, this will establish that $V^{n+1}$ is isotropic under Definition \ref{WeinbergIsotropy}. In order to construct the associated Killing fields in a neighbourhood of $q$, notice that Definition \ref{WaldIsotropy} implies via Theorem \ref{ONthm} that the space-time is a warped product of the form $I\times_fM^n_k$, with $M^n_k=(M^n,g_k)$ an $n$-dimensional Riemannian manifold of constant sectional curvature $k$. Thus, given $p\in M^n_k$ there is a small geodesic ball $B_{\epsilon}(p)$ and geodesic polar normal coordinates on $B_{\epsilon}(p)\backslash \{p\}\cong (0,\epsilon)\times S^{n-1}$, where the metric $g_k$ of $M^n_k$ has the form
\begin{align*}
g_k=dr^2+S^2_k(r)g_{\mathbb{S}^{n-1}},
\end{align*}
with $S_k$ given as in Section \ref{Preliminaries}. Therefore, the (smooth) isometries $\phi$ of $\mathbb{S}^{n-1}$ give rise to (smooth) space-time isometries of $g=-dt^2+f^2(t)g_k$ in a neighbourhood of $q=(t,p)$, given by $\varphi:I\times(0,\epsilon)\times \mathbb{S}^{n-1}$, $\varphi(t,r,s)=(t,r,\phi(s))$. Since the isometry group of $\mathbb{S}^{n-1}$ is $O(n)$,\footnote{That is, the orthogonal transformations of $\mathbb{E}^{n}$ restricted to $\mathbb{S}^{n-1}$.} we see that for each element $\mathcal{R}\in O(n)$, there is a space-time isometry of $I\times_fM^n_k$ around $q$ such that $\varphi|_{M^n_k}=\mathcal{R}$. Thus, if we consider an arbitrary smooth 1-parameter subgroup of rotations $\mathcal{R}_{\theta}$ and denote by $\varphi_{\theta}$ the associated smooth 1-parameter group of space-time isometries, we can then denote by $X_{\theta}$ be the associated (local) Killing field. Since $\varphi_{\theta}(q)=q$ $\forall \theta$, then $X_{\theta}(q)=\frac{d\varphi_{\theta}}{d\theta}(q)=0$ and $X_{\theta}\in \overline{\iota_q(\mathcal{U}_q)}$. To check the second condition in Definition \ref{WeinbergIsotropy}, we need to compute $(\nabla X\big)^{\top}\vert_{\pi(q)}$. Given $V\in \Gamma(TM)$, and denoting by $\bar{V}$ an extension to $\mathcal{U}_q$, we can do this appealing to Proposition \ref{RotationalKillingsProp}, and we get% With this in mind, consider the pull-back $\varphi_{\theta}^{*}$ acting on vector fields $Y\in \mathfrak{X}(V)$ as 
\begin{align}\label{IsotropyEquiv.01}
(\nabla_{\bar{V}_q}X)^{\top}&=\left(\frac{d(d\varphi_{\theta})_q(\bar{V}_q)}{d\theta}\Big\vert_{\theta=0}\right)^{\top}=\frac{d{\mathcal{R}_{\theta}}(V_p)}{d\theta}\Big\vert_{\theta=0}
\end{align}
%Thus, taking coordinates $\{x^{\alpha}\}_{\alpha=0}^n$ in a neighbourhood of $q$, where the coordinates $\{x^j\}_{j=1}^n$ are adapted to $M$, one sees that
%\begin{align*}
%g_q\left(\partial_j,\frac{d((d\varphi_{\theta})_q(\partial_i))}{d\theta}\Big\vert_{\theta=0}\right)=g_q(\partial_j,\nabla_{\partial_i}X)=\nabla_iX_j\big\vert_q. 
%\end{align*}
%That is,
%\begin{align}\label{IsotropyEquiv.1}
%\nabla_{i}X_j\big\vert_q=\frac{d(d\varphi_{\theta})^j_i}{d\theta}\Big\vert_{q;\theta=0}=\frac{d{\mathcal{R}_{\theta}}^j_i}{d\theta}\Big\vert_{0;\theta=0},
%\end{align}
where by construction the right-hand side in the above expression stands for an arbitrary element of $\mathfrak{o}(n)$, which proves that one can construct Killing fields $X_{\theta}$ from $\varphi_{\theta}$ which satisfy Definition \ref{WeinbergIsotropy}.

\medskip
Now, to see that Definition \ref{WeinbergIsotropy} implies Definition \ref{WaldIsotropy}, let $X$ be a Killing field satisfying Definition \ref{WeinbergIsotropy} at some fixed but arbitrary $q\in V^{n+1}$, and denote by $\varphi^{X}_s$ the associated 1-parameter family of local isometries. Notice that $X_q=0$ implies that $\varphi_s(q)=q$ due to uniqueness of integral curves $\gamma_s$ through $q$ with initial condition $\varphi_0(q)=\gamma_0=q$. Furthermore, being the Killing vector fields $X$ tangent to $U^{\perp}$, one sees that $\varphi^X_s=\mathrm{Id}\times \phi_s^X$, where $\phi_s^X=\pi\circ \varphi_s^X$ and $\pi:V^{n+1}\to M^n$ is the canonical projection, which implies that $(d\varphi^X_s)_q(U_q)=U_q$. Furthermore, denoting by $r=(t,m)\in \mathcal{U}_q$, then $X_t^{\top}(m)\doteq d\pi_{r}(X(r))=X(r)$ defines a tangent vector field induced on $\mathcal{V}_{p}\subset M^n_t$, $p=d\pi(q)$. Then, let $Y,Z\in \Gamma(TM^n_t)$, denote by $\bar{Y},\bar{Z}$ extensions to a neighbourhood of $M^n_t$ in $V^{n+1}$, notice that $X$ is defined in such a neighbourhood inducing $X^{\top}_t$ on $M^n_t$, and denote by $D:\Gamma(TM^n_t)\times \Gamma(TM^n_t)\to \Gamma(TM^n_t)$ the induced Riemannian connection on $M^n_t$ by $\nabla$. With these notations, we get
\begin{align*}
\pounds_{X_t^{\top}}g_t(Y,Z)&=g_t(D_YX_t^{\top},Z)+g_t(Y,D_ZX_t^{\top}),\\
&=g(\nabla_{\bar{Y}}X,\bar{Z})+g(\bar{Y},\nabla_{\bar{Z}}X),\\
&=\pounds_{X}g(\bar{Y},\bar{Z})=0.
\end{align*}
That is, $X^{\top}_t$ is a Killing vector field of $(M^n_t,g_t)$ in a neighbourhood $\mathcal{V}_p\subset M^n_t$ of $p=\pi(q)$, which actually satisfies $X^{\top}_t\in \iota_p(\mathcal{V}_{p})$ and whose associated local flow is given by $\overline{\phi}^X_s=\phi^X_s\circ i$, where $i:\mathcal{V}_{p}\to \mathcal{U}_q$ denotes the canonical inclusion. Since by hypotheses we can generate $\frac{n(n-1)}{2}$ such linearly independent fields on $\mathcal{V}_{p}$, then $\mathrm{dim}(\iota_p(\mathcal{V}_{p}))=\frac{n(n-1)}{2}$ and then Lemma  \ref{SIuniqueness} implies that 
\begin{align}\label{STIEquivalenceFinal.1}
\{(d\overline{\phi}^X_s)_p \: : \: X\in \overline{\iota_p(\mathcal{U}_q)}\}=SO(T_pM^n_t),
\end{align}
Thus, since 
\begin{align}\label{STIEquivalenceFinal.2}
(d\overline{\phi}^X_s)_p(v)=(d\phi^X_s)_p(i(v))=(d\varphi^X_s)_p(i(v)), \: \forall \: v\in T_pM^n_t,
\end{align}
we see that, given any pair of vectors $v,w\in T_pM^n$, there is a \emph{local isotropy Killing field} $X$ such that the associated local isometry $\varphi_s^{X}=\mathrm{Id}\times \phi_s^{X}$ fixes $q$ and $U_q$ while it rotates $v$ into $w$.

Finally, having established (\ref{STIEquivalenceFinal.1})-(\ref{STIEquivalenceFinal.2}), the same argument as in the second part of Theorem \ref{SIuniquenessTHM} shows that given any two plains $P_1,P_2\subset T_pM_t^n$, there is a \emph{local isotropy Killing field} $X$ such that the associated local isometry $\varphi_s^{X}=\mathrm{Id}\times \phi_s^{X}$ fixes $q$ and $U_q$ while it maps $P_1$ onto $P_2$.
\end{proof}

With the above proof, we have now classified all four definitions provided in this paper into the SI and STI categories, and in particular, due to Theorem \ref{ONthm}, we know the STI definitions lead to (local) rigidity with the FLRW family. In the next section, exploring the BCCCM family described in Section \ref{BCCCM}, we shall prove that the SI category is strictly weaker than the STI one, allowing for a whole new family of space-times which are non isometric (even locally) to any FLRW space-time.

Before finishing this section, let us highlight that, although obviously many concepts close to the ones exposed above have been used in standard literature, and furthermore many of the above results seem to have been intuited, it seems clear for instance from Remark \ref{NonRigRemark1}, that some of these intuitions were misplaced. Furthermore, one could point out to other classic textbooks (such as \cite[Chapter 27), Section 27.3]{MTW}, highly insightful lecture notes (for instance \cite{CarrollLN}), and recent research papers on the topic (such as \cite{Maartens01}), where the notion of isotropy is not explicitly defined in a precise enough manner to decide whether one should have in mind an SI or STI definition. Notice that, as highlighted in Remark \ref{NonRigRemark1}, the ultimate objective around these definitions in all these references is at stake depending on this subtlety. Given the range of different presentations and claims in standard references and research papers through the years, we consider that the existence of examples such as BCCCM family and their consequences have been simply overlooked and hence that the above classification, mainly put together with the next section, provide relevant results to this area.  

\section{Conformal properties of BCCCM models}\label{RigiditySection1}

In general, it might be highly non-trivial to compare two given space-times and decide on the existence of potential (local) isometries. Nevertheless, in the case of BCCCM space-times and FLRW space-times, we shall appeal to the conformal structure of each of them. In particular, FLRW space-times are known to be conformally flat.\footnote{For the sake of completeness, we have provided such a proof in the Appendix to this paper. We highlight that, for instance, in \cite[Chapter 7]{WaldBook}, this is left as an exercise, although appealing different methods.} In view of Theorem \ref{ConformallFlatnessTHM}, this provides us with a clear-cut criterion to establish that BCCCM space-times are not even locally isometric to FLRW space-times by inspection of their Weyl tensor. We shall see below in Theorem \ref{BCCCMConfFlatThm}, that there exists a highly rigid obstruction on the curvature function $k:I\to \mathbb{R}$ for a BCCCM space-time to be conformally flat. 

With the above discussion in mind, we start by presenting the following decompositions for the curvature tensor, Ricci tensor and scalar curvature of a BCCCM space-time.

\begin{lem}\label{CurvatureSplittings}
Consider a BCCCM space-time given by $V^{n+1}_{k}\cong P_{+}\times_{S_{k}}\mathbb{S}^{n-1}$. If $X,Y,Z$ are horizontal vector fields and $U,V,W$ stand for vertical vector fields, then the curvature tensor (modulo symmetries) is given by
\begin{align}\label{CurvatureSplitting}
\begin{split}
R_g(X,Y)Z&=0,\\
R_g(X,V)Y&=S^{-1}(t,r)\nabla^2S(X,Y)V,\\
R_g(X,Y)V&=R_g(V,W)X=0,\\
R_g(V,X)W%&=S^{-1}\bar{g}(V,W)\bar{\nabla}_X\mathrm{grad} S,\\
&=Sg_{\mathbb{S}^{n-1}}(V,W)\nabla_X\mathrm{grad} S\\
R_g(V,W)U%&={}^{\mathbb{S}^{n-1}}\!R(V,W)U - S^{-2}\bar{g}(\bar{\nabla}S,\bar{\nabla}S)(\bar{g}(V,U)W-\bar{g}_{W,U}V),\\
&=R_{\mathbb{S}^{n-1}}(V,W)U + g(\nabla S,\nabla S)(g_{\mathbb{S}^{n-1}}(V,U)W-g_{\mathbb{S}^{n-1}}(W,U)V)
\end{split}
\end{align}

The above, in particular, implies the following splitting for the Ricci tensor:
\begin{align}\label{RicciSplitting}
\begin{split}
\mathrm{Ric}_{g}(X,Y)&=-(n-1)S^{-1}\nabla^2S(X,Y),\\
\mathrm{Ric}_{g}(X,V)&=0,\\
\mathrm{Ric}_{g}(V,W)%&=\mathrm{Ric}_{g_{\mathbb{S}^{n-1}}}(V,W) - g(V,W)S^{\#},\\
&=\mathrm{Ric}_{g_{\mathbb{S}^{n-1}}}(V,W)-S^2g_{\mathbb{S}^{n-1}}(V,W)S^{\#}=(n-2-S^{2}S^{\#})g_{\mathbb{S}^{n-1}}(V,W)
\end{split}
\end{align}
where
\begin{align}\label{RicciSplittingNotation}
S^{\#}&=S^{-1}\Box_{P}S + \frac{n-2}{S^{2}}g_{P}(\nabla S,\nabla S),
%&=S^{-1}(\bar{g}^{tt}\bar{\nabla}_{t}\bar{\nabla}_{t}S + \bar{g}^{rr}\bar{\nabla}_{r}\bar{\nabla}_{r}S)
\end{align}
and $\Box_{P}=-\partial_t^2+\partial^2_r$ stands for the D'Alambertian on $\mathrm{P}$. Finally, all this implies that the scalar curvature satisfies the following identities:
\begin{align}\label{WPScalCurv}
\begin{split}
R_{g}&=S^{-2}\left(R_{\mathbb{S}^{n-1}} - (n-1)(2S\Box_{P}S + (n-2)g_{P}(\nabla S,\nabla S))\right),\\
&=\frac{n-1}{S^{2}}\big( (n-2)(1 - g_{P}(\nabla S,\nabla S)) - 2S\Box_{P}S\big),
\end{split}
\end{align}
\end{lem}
\begin{proof}
The expressions given in (\ref{CurvatureSplitting})-(\ref{RicciSplitting}) follow directly from the ones in \cite[Chapter 7, Proposition 42 and Corollary 43]{ONeillBook} respectively.\footnote{In the case of (\ref{CurvatureSplitting}), notice the difference in convention sign for the curvature tensor.} To establish (\ref{WPScalCurv}), Take $\{E_{\alpha}\}_{\alpha=0}^n$ a local orthonormal frame, with $E_{0}=\partial_t,E_1=\partial_r$. Then, setting $\epsilon_{\alpha}\doteq g(E_{\alpha},E_{\alpha})$:
\begin{align*}
R_{g}&=\sum_{\alpha=0}^{n}\epsilon_{\alpha}\mathrm{Ric}_g(E_{\alpha},E_{\alpha})=-\mathrm{Ric}_g(\partial_{t},\partial_{t})+\mathrm{Ric}_g(\partial_{r},\partial_{r}) + \sum_{a=1}^{n-1}\mathrm{Ric}_g(E_{a},E_{a}).
\end{align*}
We can appeal to (\ref{RicciSplitting}) to rewrite the above as
\begin{align*}
R_g&=(n-1)S^{-1}(\nabla_t\nabla_tS - \nabla_r\nabla_rS) - S^2S^{\#} \sum_{a=1}^{n-1}g_{\mathbb{S}^{n-1}}(E_{a},E_{a}) + \sum_{a=1}^{n-1}\mathrm{Ric}_{\mathbb{S}^{n-1}}(E_{a},E_{a}),\\
&=S^{-2}R_{\mathbb{S}^{n-1}} -(n-1)S^{-1}\Box_{P}S - (n-1)S^{\#},\\
%&=S^{-2}R_{\mathbb{S}^{n-1}} - (n-1)S^{-1}\Box_{P}S - (n-1)(S^{-1}\Box_{P}S + \frac{n-2}{S^{2}}g_{P}(\bar{\nabla}S,\bar{\nabla}S)),\\
&=S^{-2}\left(R_{\mathbb{S}^{n-1}} - (n-1)(2S\Box_{P}S + (n-2)g_{P}(\nabla S,\nabla S))\right),
\end{align*}
which establishes the first line in (\ref{WPScalCurv}). The second one follows simply by noticing that $R_{\mathbb{S}^{n-1}}=(n-1)(n-2)$.
%\begin{align*}
%R_{g}&=(n-1)S^{-1}(\bar{\nabla}_t\bar{\nabla}_tS - \bar{\nabla}_r\bar{\nabla}_rS) + (n-2-S^{2}S^{\#})\underbrace{\sum_{a=1}^{n-1}g_{\mathbb{S}^{n-1}}(E_{a},E_{a})}_{=\frac{n-1}{S^{2}}},\\
%&=(n-1)(-S^{-1}\Box_{P}S -S^{\#} + \frac{n-2}{S^{2}}),\\
%&=(n-1)(-S^{-1}\Box_{P}S - S^{-1}\Box_{P}S - \frac{n-2}{S^{2}}g_{P}(\bar{\nabla}S,\bar{\nabla}S) + \frac{n-2}{S^{2}}),\\
%&=\frac{n-1}{S^{2}}\big( (n-2)(1 - g_{P}(\bar{\nabla}S,\bar{\nabla}S)) - 2S\Box_{P}S\big),
%\end{align*}
\end{proof} 

With the above lemma at hand, we can know establish the following decomposition for the Weyl tensor:
\begin{lem}\label{WeylSplittingLemma}
Consider a BCCCM space-time given by $V^{n+1}_{k}\cong P_{+}\times_{S_{k}}\mathbb{S}^{n-1}$. If $X_1,\cdots,X_4$ are horizontal vector fields and $V_1,\cdots,V_4$ stand for vertical vector fields, then the Weyl tensor (modulo symmetries) is given by
\begin{align}\label{WeylSplitting}
\begin{split}
W(X_1,X_2,X_3,X_4)&=\frac{n-2}{n}S^{-2}\left(g_{P}(\nabla S,\nabla S) - 1 - S\Box_{P}S \right)(X^{\flat}_1\wedge X^{\flat}_2)\otimes (X^{\flat}_3\wedge X^{\flat}_4),\\
W(X_1,X_2,X_3,V)&=0,\\
W(X_1,X_2,V_1,V_2)&=0,\\
W(X_1,V_1,X_2,V_2)&=\frac{(n-2)}{n(n-1)}\left( g_{P}(\nabla S,\nabla S)) -1 - S\Box_PS\right)g_{\mathbb{S}^{n-1}}(V_1,V_2)g(X_1,X_2),\\
W(X_1,V_1,V_2,V_3)&=0,\\
W(V_1,V_2,V_3,V_4)%&=\frac{2S^2(g_P(\nabla S,\nabla S) - S\Box_PS -1)}{n(n-1)}(g_{\mathbb{S}^{n-1}}(V_1,V_4)g_{\mathbb{S}^{n-1}}(V_2,V_3) - g_{\mathbb{S}^{n-1}}(V_2,V_4)g_{\mathbb{S}^{n-1}}(V_1,V_3)),\\
&=-\frac{S^2(g_P(\nabla S,\nabla S) - S\Box_PS -1)}{n(n-1)}g_{\mathbb{S}^{n-1}}\KN g_{\mathbb{S}^{n-1}}(V_1,V_2,V_3,V_4),
\end{split}
\end{align}
where we have denoted by $X_i^{\flat}$ the 1-form metrically equivalent to $X_i$ and by  $X^{\flat}_i\wedge X^{\flat}_j$ the exterior product of the two 1-forms $X^{\flat}_i$ and $X^{\flat}_j$.
\end{lem}
\begin{proof}
Appealing to (\ref{WeylTensor}) and the Lemma \ref{CurvatureSplittings}, we start computing the first identity in (\ref{WeylSplitting}). First, notice that at each point the space of horizontal vectors is two dimensional. Thus, we can expand $\{X_i\}_{i=1}^4$ in an orthonormal basis for these two dimensional spaces, such us $\{\partial_t,\partial_r\}$, so that 
\begin{align*}
X_i&=X_i^{t}\partial_t + X_i^{r}\partial_r.
\end{align*}
Using these expansions, one can rewrite 
\begin{align*}
W(X_1,X_2,X_3,X_4)%&=W(X_1^{t}\partial_t + X_1^{r}\partial_r,X_2^{t}\partial_t + X_2^{r}\partial_r,X_3,X_4),\\
%&=X_1^{t}X_2^{r}W(\partial_t,\partial_r,X_3,X_4) - X_1^{r}X_2^{t}W(\partial_t,\partial_r ,X_3,X_4),\\
%&=(X_1^{t}X_2^{r} - X_1^{r}X_2^{t})W(\partial_t,\partial_r ,X_3^{t}\partial_t + X_3^{r}\partial_r,X_4^{t}\partial_t + X_4^{r}\partial_r),\\
%&=(X_1^{t}X_2^{r} - X_1^{r}X_2^{t})(X_3^{t}X_4^{r} - X_3^{r}X_4^{t})W(\partial_t,\partial_r ,\partial_t,\partial_r),\\
&=W(X_1^{t}\partial_t + X_1^{r}\partial_r,X_2^{t}\partial_t + X_2^{r}\partial_r,X_3^{t}\partial_t + X_3^{r}\partial_r,X_4^{t}\partial_t + X_4^{r}\partial_r),\\
&=(X^{\flat}_1\wedge X^{\flat}_2)\otimes (X^{\flat}_3\wedge X^{\flat}_4)W(\partial_t,\partial_r ,\partial_t,\partial_r),
\end{align*}
Then, appealing to (\ref{WeylTensor}) and the Lemma \ref{CurvatureSplittings}, we can compute
\begin{align*}
W(\partial_t,\partial_r,\partial_t,\partial_r)%&=g(\underbrace{R(\partial_t,\partial_r)\partial_r}_{=0},\partial_t) \\
%&- \frac{1}{n-1}(\mathrm{Ric}(\partial_t,\partial_t)g(\partial_r,\partial_r) + \mathrm{Ric}(\partial_r,\partial_r)g(\partial_t,\partial_t) - \mathrm{Ric}(\partial_t,\partial_r)g(\partial_r,\partial_t) - \mathrm{Ric}(\partial_r,\partial_t)g(\partial_t,\partial_r)) \\
%&- \frac{R_{g}}{n(n-1)}(g(\partial_t,\partial_r)g(\partial_r,\partial_t) - g(\partial_r,\partial_r)g(\partial_t,\partial_t)),\\
&= - \frac{1}{n-1}(\mathrm{Ric}(\partial_t,\partial_t)g(\partial_r,\partial_r) + \mathrm{Ric}(\partial_r,\partial_r)g(\partial_t,\partial_t)) + \frac{R_{g}}{n(n-1)}g(\partial_r,\partial_r)g(\partial_t,\partial_t),\\
&= S^{-1}(\nabla^2 S(\partial_t,\partial_t) - \nabla^2 S(\partial_r,\partial_r)) - \frac{R_{g}}{n(n-1)},\\
&=-S^{-1}\Box_{P}S - \frac{S^{-2}}{n}((n-2)(1-g_{P}(\nabla S,\nabla S))-2S\Box_{P}S),\\
%&=-S^{-2}\left(S\Box_{P}S -\frac{2}{n}\Box_{P}S + \frac{n-2}{n}(1-g_{P}(\nabla S,\nabla S)) \right),\\
%&=-S^{-2}\left((\frac{n-2}{n})\Box_{P}S + \frac{n-2}{n}(1-g_{P}(\nabla S,\nabla S)) \right),\\
&=\frac{n-2}{n}S^{-2}\left(g_{P}(\nabla S,\nabla S) - 1 - S\Box_{P}S \right).
\end{align*}
Therefore, we find
\begin{align}
W(X_1,X_2,X_3,X_4)=\frac{n-2}{n}S^{-2}\left(g_{P}(\nabla S,\nabla S) - 1 - S\Box_{P}S \right)(X^{\flat}_1\wedge X^{\flat}_2)\otimes (X^{\flat}_3\wedge X^{\flat}_4)
\end{align}

Concerning the second identity in (\ref{WeylSplitting}), using the orthogonality between horizontal and vertical vector fields, we have that
\begin{align*}
W(X_1,X_2,X_3,V)&=g(R(X_3,V)X_2,X_1) - \frac{1}{n-1}(\mathrm{Ric}(X_2,V)g(X_1,X_3) - \mathrm{Ric}(X_1,V)g(X_2,X_3)),
\end{align*}
Applying Lemma \ref{CurvatureSplittings} the Ricci terms are seen to vanish, and the first term gets transformed into
\begin{align}
W(X_1,X_2,X_3,V)&=S^{-1}\nabla^2S(X_2,X_3)g(V,X_1)=0
\end{align}

For the third identity in (\ref{WeylSplitting}), appealing to the same properties as above, we have that
\begin{align}
W(X_1,X_2,V_1,V_2)%&=g(R(V_1,V_2)X_2,X_1) \\
%&- \frac{1}{n-1}(\mathrm{Ric}(X_1,V_1)g(X_2,V_2) + \mathrm{Ric}(X_2,V_2)g(X_1,V_1) - \mathrm{Ric}(X_1,V_2)g(X_2,V_1) - \mathrm{Ric}(X_2,V_1)g(X_1,V_2)) \\
%&- \frac{R_{g}}{n(n-1)}(g(X_1,V_2)g(X_2,V_1) - g(X_2,V_2)g(X_1,V_1)),\\
&=g(R(V_1,V_2)X_2,X_1)=0,
\end{align}
where the last identity follows from the third identity in (\ref{CurvatureSplitting}).

For the fourth identity, using Lemma \ref{CurvatureSplittings} once more, we get
\begin{align*}
W(X_1,V_1,X_2,V_2)&=g(R(X_2,V_2)V_1,X_1) - \frac{1}{n-1}(\mathrm{Ric}(X_1,X_2)g(V_1,V_2) + \mathrm{Ric}(V_1,V_2)g(X_1,X_2)) \\
&+ \frac{R_{g}}{n(n-1)} g(V_1,V_2)g(X_1,X_2),\\
%&=-Sg_{\mathbb{S}^{n-1}}(V_1,V_2)g(\nabla_{X_2}\mathrm{grad}S,X_1) \\
%&- \frac{1}{n-1}( -(n-1)S^{-1}\nabla^2 S(X_1,X_2)g(V_1,V_2) + (n-2-S^2S^{\#})g_{\mathbb{S}^{n-1}}(V_1,V_2)g(X_1,X_2) ) \\
%&+ \frac{R_{g}}{n(n-1)}g(V_1,V_2)g(X_1,X_2),\\
&=-Sg_{\mathbb{S}^{n-1}}(V_1,V_2)g(\nabla_{X_2}\mathrm{grad}S,X_1) + S^{-1}g(V_1,V_2)\nabla^2 S(X_1,X_2) \\
& - \frac{1}{n-1}(n-2-SS^{\#})g_{\mathbb{S}^{n-1}}(V_1,V_2)g(X_1,X_2)  \\
&+ \frac{S^{-2}\big((n-2)(1-g_{P}(\nabla S,\nabla S))-2S\Box_PS\big)}{n}g(V_1,V_2)g(X_1,X_2),\\
%&=\left((n-2)\left(\frac{1}{n} -\frac{1}{n-1}\right)(1-g_{P}(\nabla S,\nabla S)) + \left(\frac{1}{n-1} - \frac{2}{n}\right)S\Box_PS\right)g_{\mathbb{S}^{n-1}}(V_1,V_2)g(X_1,X_2),\\
&=\frac{(n-2)}{n(n-1)}\left( g_{P}(\nabla S,\nabla S)) -1 - S\Box_PS\right)g_{\mathbb{S}^{n-1}}(V_1,V_2)g(X_1,X_2)
\end{align*}
where we have used that $S^2S^{\#}=S\Box_PS+(n-2)g_{P}(\nabla S,\nabla S)$ (which follows from (\ref{RicciSplittingNotation})), and the first line in the second identity vanishes from the definition of the Hessian.

Concerning the fifth identity in (\ref{WeylSplitting}), appealing again to the orthogonality of horizontal and vertical vectors, as well as the second identity in (\ref{RicciSplitting}), we have
\begin{align*}
W(X_1,V_1,V_2,V_3)%&=g(R(V_2,V_3)V_1,X_1) \\
%&- \frac{1}{n-1}(\mathrm{Ric}(X_1,V_2)g(V_1,V_3) + \mathrm{Ric}(V_1,V_3)g(X_1,V_2) - \mathrm{Ric}(X_1,V_3)g(V_1,V_2) - \mathrm{Ric}(V_1,V_2)g(X_1,V_3)) \\
%&- \frac{R_{g}}{n(n-1)}(g(X_1,V_3)g(V_1,V_2) - g(V_1,V_3)g(X_1,V_2)),\\
&=g(R(V_2,V_3)V_1,X_1)=0,
\end{align*}
where the last identity follows from (\ref{CurvatureSplitting}) since $R(V_2,V_3)V_1$ is vertical.

Finally, for the last identity in (\ref{WeylSplitting}), we have
\begin{align*}
W(V_1,V_2,V_3,V_4)&=g(R(V_3,V_4)V_2,V_1) -\frac{1}{n-1}\mathrm{Ric}_g\KN g(V_1,V_2,V_3,V_4) +\frac{R_g}{2n(n-1)}g\KN g(V_1,V_2,V_3,V_4)\\
%&=g(R(V_3,V_4)V_2,V_1) - \frac{1}{n-1}\Big(\mathrm{Ric}(V_1,V_3)g(V_2,V_4) + \mathrm{Ric}(V_2,V_4)g(V_1,V_3) \\
%&- \mathrm{Ric}(V_1,V_4)g(V_2,V_3) - \mathrm{Ric}(V_2,V_3)g(V_1,V_4)\Big) \\
%&- \frac{R_{g}}{n(n-1)}(g(V_1,V_4)g(V_2,V_3) - g(V_2,V_4)g(V_1,V_3)),\\
&=S^2g_{\mathbb{S}^{n-1}}(R_{\mathbb{S}^{n-1}}(V_3,V_4)V_2,V_1) - \frac{S^2g_{P}(\nabla S,\nabla S)}{2}g_{\mathbb{S}^{n-1}}\KN g_{\mathbb{S}^{n-1}}(V_1,V_2,V_3,V_4) \\
& - \frac{S^2(n-2-S^2S^{\#})}{n-1}g_{\mathbb{S}^{n-1}}\KN g_{\mathbb{S}^{n-1}}(V_1,V_2,V_3,V_4) \\
&+\frac{S^{2}\big((n-2)(1-g_{P}(\nabla S,\nabla S)) -2S\Box_PS\big)}{2n}g_{\mathbb{S}^{n-1}}\KN g_{\mathbb{S}^{n-1}}(V_1,V_2,V_3,V_4)\\
%&=S^2g_{\mathbb{S}^{n-1}}(R_{\mathbb{S}^{n-1}}(V_3,V_4)V_2,V_1) \\
%&- S^2\Big(\frac{g_{P}(\nabla S,\nabla S)}{2} + \frac{(n-2-S^2S^{\#})}{n-1} - \frac{(n-2)(1-g_{P}(\nabla S,\nabla S)) -2S\Box_PS}{2n}\Big)g_{\mathbb{S}^{n-1}}\KN g_{\mathbb{S}^{n-1}}(V_1,V_2,V_3,V_4)\\
&=S^2g_{\mathbb{S}^{n-1}}(R_{\mathbb{S}^{n-1}}(V_3,V_4)V_2,V_1) - \frac{S^2}{2}hg_{\mathbb{S}^{n-1}}\KN g_{\mathbb{S}^{n-1}}(V_1,V_2,V_3,V_4)
\end{align*}
where we have defined
\begin{align}\label{WeylSplitting.1}
h\doteq g_{P}(\nabla S,\nabla S) + \frac{2(n-2-S^{2}S^{\#})}{n-1} + \frac{2S\Box_PS + (n-2)g_{P}(\nabla S,\nabla S)  + 2-n}{n}
\end{align}
Furthermore, since $\mathbb{S}^{n-1}$ is a space of constant sectional curvature equal to one, we also know that
\begin{align*}
g_{\mathbb{S}^{n-1}}(R_{\mathbb{S}^{n-1}}(V_3,V_4)V_2,V_1)%&=g_{\mathbb{S}^{n-1}}(V_2,V_4)g_{\mathbb{S}^{n-1}}(V_1,V_3) - g_{\mathbb{S}^{n-1}}(V_3,V_2)g_{\mathbb{S}^{n-1}}(V_1,V_4),\\
&=\frac{1}{2}g_{\mathbb{S}^{n-1}}\KN g_{\mathbb{S}^{n-1}}(V_1,V_2,V_3,V_4)
\end{align*}
and therefore
\begin{align*}
W(V_1,V_2,V_3,V_4)&=\frac{S^2}{2}(1 - h)g_{\mathbb{S}^{n-1}}\KN g_{\mathbb{S}^{n-1}}(V_1,V_2,V_3,V_4).
\end{align*}

Using (\ref{WeylSplitting.1}), we can compute that
\begin{align*}
h-1%&=g_{P}(\nabla S,\nabla S) - \frac{2S^{2}S^{\#}}{n-1} - \frac{R_{\mathbb{S}^{n-1}}}{n(n-1)}+ \frac{(2S\Box_PS + (n-2)g_{P}(\nabla S,\nabla S))}{n}+\frac{2(n-2)}{n-1}-1,\\
%&=\left(1 + \frac{(n-2)}{n} - \frac{2(n-2)}{n-1}\right)g_P(\nabla S,\nabla S)  + (\frac{1 }{n} - \frac{1}{n-1})2S\Box_P - \frac{n-2}{n}+\frac{2(n-2)}{n-1}-1,\\
&=\frac{2}{n(n-1)}(g_P(\nabla S,\nabla S)  - S\Box_PS -1),
\end{align*}
which finally implies that
\begin{align}
W(V_1,V_2,V_3,V_4)&=-\frac{S^2(g_P(\nabla S,\nabla S) - S\Box_PS -1)}{n(n-1)}g_{\mathbb{S}^{n-1}}\KN g_{\mathbb{S}^{n-1}}(V_1,V_2,V_3,V_4)
\end{align}
\end{proof}

The above lemma gives us the following direct characterisation for \emph{conformal flatness} of the space-times $V_{k(t)}=P_{+}\times_{S_{k(t)}}\mathbb{S}^{n-1}$, which we shall be further exploited below.

%It is a well-known fact that a semi-Riemannian manifold of dimension $d\geq 4$ is conformally flat iff its Weyl tensor vanishes \cite[Theorem 1.165]{Besse},\cite{Eis}. Therefore, we can now establish the following corollary of Lemma \ref{WeylSplittingLemma}.

\begin{cor}\label{CondFlatCoro}
Consider a BCCCM space-time given by $V^{n+1}_{k}\cong P_{+}\times_{S_{k}}\mathbb{S}^{n-1}$. Then, $V^{n+1}_{k}$ is conformally flat iff
\begin{align}\label{ConfFlatCondition}
S\partial^2_tS=\left(\partial_tS\right)^2.
\end{align}
\end{cor}
\begin{proof}
From Lemma \ref{WeylSplittingLemma}, one has that $V_{k(t)}$ is conformally flat iff
\begin{align}\label{ConfFlat.1}
g_P(\nabla S,\nabla S) - S\Box_PS -1=0.
\end{align}

We can rewrite (\ref{ConfFlat.1}) noticing that $dS=-\partial_tSdt + \partial_rSdr$, which implies
\begin{align*}
g_P(\nabla S,\nabla S)=-(\partial_tS)^2 + (\partial_rS)^2,\;\; S\Box_PS=-S\partial_t^2S + S\partial^2_rS.
\end{align*}
But also, from its definition, we know that $S\partial^2_rS=-k(t)S^2$ and that $(\partial_rS)^2+kS^2=1$ (see Section \ref{BCCCM}). Therefore,
\begin{align*}
g_P(\nabla S,\nabla S) - S\Box_PS -1&=-(\partial_tS)^2 + (\partial_rS)^2 +S\partial_t^2S - S\partial^2_rS-1,\\
&=-(\partial_tS)^2 + S\partial_t^2S  + (\partial_rS)^2+ kS^2-1,\\
&=-(\partial_tS)^2 + S\partial_t^2S.
\end{align*}
The above together with (\ref{ConfFlat.1}) implies our claim.
\end{proof}

The above corollary can be now used to explicitly show that (generically) the BCCCM family of cosmological space-times is not conformally flat. This follows from the following theorem.

\begin{theo}\label{BCCCMConfFlatThm}
A BCCCM space-time given by $V^{n+1}_{k}\cong P_{+}\times_{S_{k}}\mathbb{S}^{n-1}$ is conformally flat iff $k:I\subset \mathbb{R}\to \mathbb{R}$ is a constant function. 
\end{theo}
\begin{proof}
Clearly, if the sectional curvature $k$ is a constant function, then (\ref{ConfFlatCondition}) is satisfied and the result follows from Corollary \ref{CondFlatCoro}. To see the converse, let us assume that (\ref{ConfFlatCondition}) is satisfied. One can directly compute that 
\begin{align*}
\frac{\partial_tS}{S}&=\begin{cases}
\frac{1}{2}\frac{k'}{k}\left(\sqrt{-k}r\frac{\cosh(\sqrt{-k}r)}{\sinh(\sqrt{-k}r)} - 1 \right), \text{ if  } k<0,\\
-\frac{1}{6}k'r^2, \text{ if } k=0,\\
\frac{1}{2}\frac{k'}{k}\left(\sqrt{k}r\frac{\cos(\sqrt{k}r)}{\sin(\sqrt{k}r)} - 1 \right), \text{ if  } k>0
\end{cases}
\end{align*}
If $k\not\equiv cte$, then there is some $t_0$ where the functions $k,k'\neq 0$ in a neighbourhood of it. Notice that in this neighbourhood, $\partial_tS(t,r)=0$ iff
\begin{align}\label{ConfromalFlatRigidity.0}
\begin{split}
0=\begin{cases}
\tanh(\sqrt{-k(t)}r) - \sqrt{-k(t)}r , \text{ if  } k<0,\\
\tan(\sqrt{k(t)}r) - \sqrt{k(t)}r, \text{ if  } k>0
\end{cases}
\end{split}
\end{align}
In both of the above cases, the equation can only be satisfied for $\sqrt{\pm k(t)}r$ equal to distinguished fixed values. That is, we must have $\sqrt{\pm k(t)}r=c$ for some constant $c$. Since $k'(t)\neq 0$ in this interval, if $r_0\neq 0$,\footnote{Notice that for $r=0$ (\ref{ConfromalFlatRigidity.0}) is clearly satisfied for all $t$.} there must be some time $t_0$ where this equality fails at $(t_0,r_0)$, and restricting to a small enough neighbourhood of $(t_0,r_0)$ such equality must fail for all $t$ in it. Therefore, we can assume that, if $k\not\equiv cte$, for any $r_0\neq 0$, there is some $t_0\in I$ and a neighbourhood $\mathcal{I}_0$ of $(t_0,r_0)$ such that $k,k',\partial_tS\neq 0$ on it. Thus, for $(t,r)\in \mathcal{I}_0$, $t>t_0$, we can rewrite (\ref{ConfFlatCondition}) as
\begin{align*}
\partial_t\log|S'|=\partial_t\log(S),
\end{align*}
where we have denoted $S'=\partial_tS$. The above implies that
\begin{align}\label{ConformalFlatRigity.1}
\frac{S'(t,r)}{S(t,r)}=\frac{S'(t_0,r)}{S(t_0,r)}
\end{align}

Assuming that $k<0$ on $\mathcal{I}_0$ one obtains that the above is equivalent to
\begin{align*}
%\frac{k'}{k}\left(\sqrt{-k}r\frac{\cosh(\sqrt{-k}r)}{\sinh(\sqrt{-k}r)} - 1 \right)=\frac{k'_0}{k_0}\left(\sqrt{-k_0}r\frac{\cosh(\sqrt{-k_0}r)}{\sinh(\sqrt{-k_0}r)} - 1 \right)
\frac{\sqrt{-k}r\frac{\cosh(\sqrt{-k}r)}{\sinh(\sqrt{-k}r)} - 1}{\sqrt{-k_0}r\frac{\cosh(\sqrt{-k_0}r)}{\sinh(\sqrt{-k_0}r)} - 1}=\frac{k'_0}{k_0}\frac{k}{k'},
\end{align*}
where the right-hand side is independent of $r$. Let us denote this right-hand side by $u(t)$ and, for fixed $t\in \mathcal{I}_0$, let us examine the left-hand side near $r=0$ as follows\footnote{Notice that, for $t$ fixed in a neighbourhood of $t_0$, the function on the left-hand side extends to a well-defined continuous function of $r$ in a neighbourhood of $r=0$. In particular, it equals $\frac{k(t)}{k_0}$ at $r=0$.}
\begin{align*}
u(t)%&=\frac{\frac{\cosh(\sqrt{-k}r)}{\frac{\sinh(\sqrt{-k}r)}{\sqrt{-k}r}} - 1}{\frac{\cosh(\sqrt{-k_0}r)}{\frac{\sinh(\sqrt{-k_0}r)}{\sqrt{-k_0}r}} - 1},\\
&=\frac{\frac{\sinh(\sqrt{-k_0}r)}{\sqrt{-k_0}r}}{\frac{\sinh(\sqrt{-k}r)}{\sqrt{-k}r}}\frac{\cosh(\sqrt{-k}r) - \frac{\sinh(\sqrt{-k}r)}{\sqrt{-k}r}}{\cosh(\sqrt{-k_0}r) - \frac{\sinh(\sqrt{-k_0}r)}{\sqrt{-k_0}r}}
%&=\frac{1-\frac{1}{6}k_0r^2+\frac{1}{5!}k_0^2r^4+O(k_0^3r^6)}{1-\frac{1}{6}kr^2+\frac{1}{5!}k^2r^4+O(k^3r^6)}\frac{1-\frac{1}{2}kr^2+\frac{1}{24}k^2r^4+O(k^3r^6)-(1-\frac{1}{6}kr^2+\frac{1}{5!}k^2r^4+O(k^3r^6))}{1-\frac{1}{2}k_0r^2+\frac{1}{24}k^2_0r^4+O(k_0^3r^6)-(1-\frac{1}{6}k_0r^2+\frac{1}{5!}k^2_0r^4+O(k_0^3r^6))},\\
%&=\frac{(1-\frac{1}{6}k_0r^2+\frac{1}{5!}k_0^2r^4+O(k_0^3r^6))(-\frac{1}{3}kr^2 + \frac{1}{30}k^2r^4+O(k^3r^6))}{(1-\frac{1}{6}kr^2+\frac{1}{5!}k^2r^4+O(k^3r^6))(-\frac{1}{3}k_0r^2 + \frac{1}{30}k_0^2r^4+O(k_0^3r^6))},\\
%&=\frac{-\frac{1}{3}kr^2 + \frac{1}{18}k_0kr^4 + \frac{1}{30}k^2r^4 + O(r^6) }{-\frac{1}{3}k_0r^2 + \frac{1}{18}k_0kr^4 + \frac{1}{30}k^2_0r^4 +O(r^6) },\\
%&=\frac{-\frac{1}{3}kr^2 + \frac{1}{6}(\frac{k_0}{3}+\frac{k}{5})kr^4 + O(r^6) }{-\frac{1}{3}k_0r^2 + \frac{1}{6}(\frac{k}{3}+\frac{k_0}{5})k_0r^4 +O(r^6) },\\
%&=\frac{\frac{k}{k_0} - \frac{1}{2}(\frac{k_0}{3}+\frac{k}{5})\frac{k}{k_0}r^2 + O(r^4)}{1 - \frac{1}{2}(\frac{k}{3}+\frac{k_0}{5})r^2 +O(r^4)},\\
%&=\left(\frac{k}{k_0} - \frac{1}{2}(\frac{k_0}{3}+\frac{k}{5})\frac{k}{k_0}r^2 + O(r^4)\right)\left(1 + \frac{1}{2}(\frac{k}{3}+\frac{k_0}{5})r^2 +O(r^4)\right),\\
%&=\frac{k}{k_0}-\frac{1}{2}(\frac{k_0}{3}+\frac{k}{5})\frac{k}{k_0}r^2 + \frac{1}{2}(\frac{k}{3}+\frac{k_0}{5})\frac{k}{k_0}r^2 + O(r^4),\\
%&=\frac{k}{k_0} + \frac{1}{2}\left( \frac{k}{3}+\frac{k_0}{5} - \frac{k_0}{3} - \frac{k}{5}\right)\frac{k}{k_0}r^2 + O(r^4),\\
=\frac{k}{k_0} + \frac{1}{15}\left( k- k_0\right)\frac{k}{k_0}r^2 + O(r^4),
\end{align*}
where the above result is obtained by Taylor expansion of the associated hyperbolic functions up to fourth and sixth order respectively for the cosine and sine. We can then see that a necessary condition for the right-hand side in the above expression to be independent of $r$ is that the coefficient in front of the $r^2$-term must vanish. That is, we must have that $k(t)=k_0$ on $\mathcal{I}_0$, but this contradicts our initial hypothesis that $k'\neq 0$ on $\mathcal{I}_0$. We therefore conclude that no such neighbourhood $\mathcal{I}_0$ can exist for $r_0$ in a neighbourhood of $r=0$, which implies that $k$ must be a constant by our initial discussion. The same line of reasoning also works after (\ref{ConformalFlatRigity.1}) for the case $k>0$, which concludes the proof.
\end{proof}

The above theorem allows us to finally conclude that the BCCCM family of space-times $V^{n+1}_{k(t)}$ are non-isometric to any FLRW, unless they are trivial. That is:

\begin{cor}
Consider a BCCCM space-time given by $V^{n+1}_{k}\cong P_{+}\times_{S_{k}}\mathbb{S}^{n-1}$. Then, $V^{n+1}_{k(t)}$ is locally isometric to a FLRW space-time iff $k$ is constant.
\end{cor}
\begin{proof}
If $k$ is constant the result is obvious. The converse follows because FLRW space-times are known to be conformally flat (see, for instance, Theorem \ref{ConfFlatWP}), but if $k$ is not a constant, the space-times $V^{n+1}_{k(t)}$ are not conformally flat due to the previous theorem.
\end{proof}

One can use the the above results to establish that Definition \ref{CBIsotropy} is strictly weaker than Definition \ref{WaldIsotropy}, since due to Theorem \ref{ONthm}, a cosmological space-time with $M^n_t$ simply connected and $g_t$ complete which is isotropic under Definition \ref{WaldIsotropy} must be in the FLRW family. Therefore, the well-known rigidity associated to isotropic cosmological space-times is only true under Definition \ref{WaldIsotropy}.

\subsection*{Acknowledgements}
The author would like to thank professor Miguel Sánchez (University of Granada, Spain), professor Carlos A. Romero (Federal University of Paraíba, Brazil) and professor Fábio Dahia (Federal University of Paraíba, Brazil) for discussions about the object of this paper, several useful suggestions and comments, as well as the reading of a preliminary version of it. Also, the author would like to thank the Alexander von Humboldt Foundation and the Fundação Cearense de Apoio ao Desenvolvimento Científico e Tecnológico (FUNCAP) for partial financial support during the writing of this paper.

\appendix
\markboth{Appendix}{Appendix}
\addcontentsline{toc}{section}{Appendices}
\renewcommand{\thesubsection}{\Alph{subsection}}
\numberwithin{equation}{section}
\numberwithin{theo}{section}
%\numberwithin{coro}{section}
\numberwithin{remark}{section}

\section{Appendix}

\subsection*{Conformally Flat Warped-Product Space-Times}

The objective of this appendix is to prove Theorem \ref{ConfFlatWP} below. We present here this result for the sake of completeness and because a geometric explicit proof of the conformal flatness of the FLRW family is not so easy to find in the literature. %In particular, in this last theorem one does not require any a priori additional symmetry for the fibre of the warped-product spaces, as is done in the case of the FLRW family.
With this in mind, we first need the follow Lemma.

\begin{lem}\label{ConfFlatWPSTLemma}
Let $(M^n,\gamma)$ be a connected $n$-dimensional Riemannian manifold and consider the warped-product $I\times_{S}M^n$, $I\subset \mathbb{R}$ and open set, equipped with the metric 
\begin{align}
g=-dt^2+S^2(t)\gamma.
\end{align}
Then, if $U,V,W$ are vertical vector fields, the following decompositions follow for the curvature and Ricci tensor
\begin{align}\label{WPCurvatureSplittingAP}
\begin{split}
R(V,\partial_t)\partial_t&=-S^{-1}\nabla^2S(\partial_t,\partial_t)V,\\
R(\partial_t,\partial_t)V&=R(V,W)\partial_t=0,\\
R(\partial_t,V)W&=-S\gamma(V,W)\nabla_{\partial_t}\mathrm{grad}S,\\
R(V,W)U&=R_\gamma(V,W)U+|\nabla S|^2(\gamma(U,V)W-\gamma(U,W)V),
\end{split}
\end{align}
and
\begin{align}\label{WPRicciSplittingAP}
\begin{split}
\mathrm{Ric}(\partial_t,\partial_t)&=-\frac{n}{S}\nabla^2S(\partial_t,\partial_t)=-n\frac{S''}{S},\\
\mathrm{Ric}(V,\partial_t)&=0,\\
\mathrm{Ric}(V,W)&=\mathrm{Ric}_{\gamma}(V,W)-S^2S^{\#}\gamma(V,W),
\end{split}
\end{align}
where $S^{\#}=-\frac{S''}{S}-(n-1)\left( \frac{S'}{S}\right)^2$. Furthermore, the scalar curvature can be written as
\begin{align}\label{WPScalCurvSplittingAP}
R_g=S^{-2}R_{\gamma} + n(n-1)\left(\frac{S'}{S}\right)^2 + 2n\frac{S''}{S}.
\end{align}
\end{lem}
\begin{proof}
The relations (\ref{WPCurvatureSplittingAP}) and (\ref{WPRicciSplittingAP}) follow directly from \cite[Chapter 9, Proposition 42 and Corollary 43]{ONeillBook}. In the case of (\ref{WPScalCurvSplittingAP}), considering a local orthonormal frame $\{E_{\alpha}\}_{\alpha=0}^n$, with $E_0=\partial_t$, it follows that
\begin{align*}
R_g&=\sum_{\alpha=0}^n\epsilon_{\alpha}\mathrm{Ric}_g(E_{\alpha},E_{\alpha})=-\mathrm{Ric}_g(\partial_t,\partial_t) + \sum_{i=1}^n\mathrm{Ric}_g(E_i,E_i),\\
&=n\frac{S''}{S} + \sum_{i=1}^n\mathrm{Ric}_{\gamma}(E_i,E_i) - S^2S^{\#}\sum_{i=1}^n\gamma(E_i,E_i)
\end{align*}
Since $\{SE_i\}_{i=1}^n$ is a local orthonormal frame on $(M,\gamma)$, we get
\begin{align*}
R_g&=S^{-2}R_{\gamma}-nS^{\#}+n\frac{S''}{S}=S^{-2}R_{\gamma} +n(n-1)\left( \frac{S'}{S}\right)^2 + 2n\frac{S''}{S}.
\end{align*}
\end{proof}

We can now establish the following result:
\begin{theo}\label{ConfFlatWP}
Let $M^n_{\gamma}\doteq (M^n,\gamma)$ be an $n$-dimensional Riemannian manifold and consider the warped-product $I\times_{S}M^n_{\gamma}$, $I\subset \mathbb{R}$ and open set. %, equipped with the metric 
%\begin{align}
%g=-dt^2+S^2(t)\gamma.
%\end{align}
If $(M^n,\gamma)$ is Einstein and conformally flat, then $I\times_{S}M^n_{\gamma}$ is conformally flat.
\end{theo}
\begin{proof}
In order to establish this result, we need to prove that the Weyl tensor of the warped product vanishes identically. Thus, let us first consider $U,V$ to be two arbitrary vertical vector fields and compute
\begin{align*}
W(\partial_t,U,\partial_t,V)&=g(R(\partial_t,V)U,\partial_t)-\frac{1}{n-1}\mathrm{Ric}_g\KN g(\partial_t,U,\partial_t,V) + \frac{R_g}{2n(n-1)}g\KN g(\partial_t,U,\partial_t,V).
\end{align*}
Using (\ref{WPCurvatureSplittingAP}) and the orthogonoality of horizontal and vertical fields, the following relations follow:
\begin{align*}
g(R(\partial_t,V)U,\partial_t)&=-\gamma(U,V)S\nabla^2S(\partial_t,\partial_t)=-\gamma(U,V)SS'',\\
\mathrm{Ric}\KN g(\partial_t,U,\partial_t,V)&=\mathrm{Ric}_g(\partial_t,\partial_t)g(U,V)+\mathrm{Ric}_g (U,V)g(\partial_t,\partial_t)=-n\frac{S''}{S}g(U,V) - \mathrm{Ric}_g(U,V),\\
g\KN g(\partial_t,U,\partial_t,V)&=-2g(U,V).
\end{align*}
Therefore, using (\ref{WPRicciSplittingAP}) and the fact that $\mathrm{Ric}_{\gamma}=c\gamma$, for some fixed constant $c$, we find
\begin{align*}
W(\partial_t,U,\partial_t,V)%&=-\gamma(U,V)SS''+\frac{1}{n-1}(n\frac{S''}{S}g(U,V) + \mathrm{Ric}_g(U,V))- \frac{R_g}{n(n-1)}g(U,V),\\
&=-\gamma(U,V)SS''+\frac{1}{n-1}(nSS''\gamma(U,V) + \mathrm{Ric}_{\gamma}(U,V) - S^2S^{\#}\gamma(U,V) )- \frac{S^{2}R_g}{n(n-1)}\gamma(U,V),\\
%&=\left(-SS''+\frac{n}{n-1}SS'' + \frac{c}{n-1} - \frac{1}{n-1}S^2S^{\#} - \frac{S^{2}R_g}{n(n-1)}\right)\gamma(U,V),\\
%&=\left(-SS''+\frac{n}{n-1}SS'' + \frac{c}{n-1} - \frac{1}{n-1}S^2S^{\#} - \frac{S^{2}R_g}{n(n-1)}\right)\gamma(U,V),\\
%&=\left(\frac{1}{n-1}SS'' + \frac{c}{n-1} - \frac{1}{n-1}S^2S^{\#} - \frac{S^{2}R_g}{n(n-1)}\right)\gamma(U,V),\\
%&=\frac{1}{n-1}\left(SS'' + c - S^2S^{\#} - \frac{cn+n(n-1)(S')^2+2nSS''}{n}\right)\gamma(U,V),\\
&=\frac{1}{n-1}\left( - S^2S^{\#} -(n-1)(S')^2 -SS''\right)\gamma(U,V).
\end{align*}
Using the the definition of $S^{\#}$, we see that the factor between brackets in the last line of the above expression vanishes identically, so
\begin{align}\label{WPConfFlatAP.0}
W(\partial_t,U,\partial_t,V)=0
\end{align}

Let us now consider $U,V$ and $W$ to be vertical vector fields and analyse the following expression:
\begin{align*}
W(\partial_t,U,V,W)&=g(R(V,W)U,\partial_t) - \frac{1}{n-1}\mathrm{Ric}_g\KN g(\partial_t,U,V,W) +\frac{R_g}{2n(n-1)}g\KN g(\partial_t,U,V,W).
\end{align*}
Using (\ref{WPRicciSplittingAP}) one can see that the second term above vanishes, while the orthogonality between horizontal and vertical vectors implies the vanishing of the $g\KN g$-term. Finally, noticing that $R(V,W)U$ is also vertical, we see that the first term vanishes as well, implying:
\begin{align}\label{WPConfFlatAP.1}
W(\partial_t,U,V,W)&=0.
\end{align}

Now consider $V_1,\cdots,V_4$ vertical vector fields and compute 
\begin{align*}
W(V_1,V_2,V_3,V_4)&=g(R(V_3,V_4)V_2,V_1) -\frac{1}{n-1}\mathrm{Ric}_g\KN g(V_1,V_2,V_3,V_4) + \frac{R_g}{2n(n-1)}g\KN g(V_1,V_2,V_3,V_4),\\
&=g(R_{\gamma}(V_3,V_4)V_2,V_1) + g(\nabla S,\nabla S)(\gamma(V_2,V_3)g(V_4,V_1) - \gamma(V_2,V_4)g(V_3,V_1)) \\
&-\frac{1}{n-1}\left(\mathrm{Ric}_{\gamma}\KN g(V_1,V_2,V_3,V_4) - S^2S^{\#}\gamma\KN g(V_1,V_2,V_3,V_4)\right) \\
&+\frac{R_g}{2n(n-1)}g\KN g(V_1,V_2,V_3,V_4),\\
%&=g(R_{\gamma}(V_3,V_4)V_2,V_1) + \frac{S^2(S')^2}{2}\gamma\KN \gamma(V_1,V_2,V_3,V_4) -\frac{S^2\left(c - S^2S^{\#}\right)}{n-1}\gamma\KN \gamma(V_1,V_2,V_3,V_4) \\
%&+ \frac{S^{4}R_g}{2n(n-1)}\gamma\KN \gamma(V_1,V_2,V_3,V_4),\\
%&=g(R_{\gamma}(V_3,V_4)V_2,V_1) + S^2\left(\frac{(S')^2}{2} - \frac{\left(c - S^2S^{\#}\right)}{n-1} + \frac{S^2R_g}{2n(n-1)} \right)\gamma\KN \gamma(V_1,V_2,V_3,V_4),\\
%&=g(R_{\gamma}(V_3,V_4)V_2,V_1)\\
%& + \left(\frac{S^2(S')^2}{2} - \frac{S^2\left(c - S^2S^{\#}\right)}{n-1} + \frac{S^2(cn + n(n-1)(S')^2 + 2nSS'')}{2n(n-1)} \right)\gamma\KN \gamma(V_1,V_2,V_3,V_4),\\
%&=g(R_{\gamma}(V_3,V_4)V_2,V_1)\\
%& + S^2\left(\frac{(S')^2}{2} - \frac{1}{n-1}\left(c - S^{2}S^{\#} - \frac{c + (n-1)(S')^2 + 2SS'')}{2} \right)\right)\gamma\KN \gamma(V_1,V_2,V_3,V_4),\\
%&=g(R_{\gamma}(V_3,V_4)V_2,V_1)\\
%& + S^2\left(\frac{(S')^2}{2} - \frac{1}{n-1}\left(\frac{c}{2} + SS'' + (n-1)(S')^{2} - \frac{(n-1)(S')^2 + 2SS'')}{2} \right)\right)\gamma\KN \gamma(V_1,V_2,V_3,V_4),\\
&=S^2\gamma(R_{\gamma}(V_3,V_4)V_2,V_1) - \frac{S^2c}{2(n-1)}\gamma\KN \gamma(V_1,V_2,V_3,V_4).
\end{align*}
One can directly see that $W_{\gamma}=0$ and $\mathrm{Ric}_{\gamma}=c\gamma$ imply that the last line in the above expression vanishes identically.
%\begin{align*}
%S^{-2}g(R_{\gamma}(V_3,V_4)V_2,V_1)&=  \frac{1}{n-2}\left(c - \frac{nc}{2(n-1)}\right)\gamma\KN \gamma(V_1,V_2,V_3,V_4)=\frac{c}{n-2}\frac{n-2}{2(n-1)}\gamma\KN \gamma(V_1,V_2,V_3,V_4),\\
%&=\frac{c}{2(n-1)}\gamma\KN \gamma(V_1,V_2,V_3,V_4)
%\end{align*}
Therefore
\begin{align}\label{WPConfFlatAP.2}
W(V_1,V_2,V_3,V_4)&=0.
\end{align}
Finally, the symmetries of the Weyl tensor put together with (\ref{WPConfFlatAP.0}),(\ref{WPConfFlatAP.1}) with (\ref{WPConfFlatAP.2}) imply that $W_g=0$, and therefore the theorem follows.
\end{proof}

Finally, let us just notice that the class of manifolds considered above is precisely the one of relevance in the study of isotropic space-times (in particular for the FLRW family). This follows in view of Theorem \ref{ONthm} and the following simple proposition below:

\begin{prop}
A connected Einstein semi-Riemannian manifold $(M^n,g)$ is conformally-flat iff it has constant sectional curvature.
\end{prop}
\begin{proof}
From (\ref{WeylTensor}), we see that the Einstein condition implies
\begin{align*}
\begin{split}
W(V,X,Y,Z)%&=g(R(Y,Z)X,V) - \frac{R_g}{n(n-1)}g\KN g (V,X,Y,Z) + \frac{R_{g}}{2n(n-1)}g\KN g(V,X,Y,Z),\\
&=g(R(Y,Z)X,V) - \frac{R_g}{2n(n-1)}g\KN g (V,X,Y,Z).
\end{split}
\end{align*}
Therefore, $(M^n,g)$ is conformally flat iff
\begin{align*}
R(V,X,Y,Z)=g(R(Y,Z)X,V) = \frac{R_g}{2n(n-1)}g\KN g (V,X,Y,Z)
\end{align*}
Therefore, given any point $p\in M$ and any non-degenerate plane $P$ generated by two tangent vectors $\{v,w\}$, we see that
\begin{align*}
k(P)=\frac{R_g}{n(n-1)},
\end{align*}
which is constant from connectedness. 
\end{proof}

\section*{Declarations}

\subsection*{Funding}

The author would like to thank the Alexander von Humboldt Foundation and the Fundação Cearense de Apoio ao Desenvolvimento Científico e Tecnológico (FUNCAP) for partial financial support during the writing of this paper.

\subsection*{Data availability statement}

Data sharing is not applicable to this article as no datasets were generated or analysed during the current study.

\subsection*{Conflict of interest}

The corresponding author states that there is no conflict of interest.

\addcontentsline{toc}{section}{References}
%\bibliographystyle{unsrt}
%\bibliography{bibliography}
\printbibliography

\end{document}